\documentclass[]{article}

\usepackage{hyperref}

\usepackage{graphicx,psfrag, xcolor}
\usepackage{subcaption}
\usepackage{amsmath}
\usepackage{amsthm}
\usepackage{cite}
\usepackage{latexsym}
\usepackage{amssymb}
\usepackage{graphics}
\usepackage{dsfont}
\usepackage{mathrsfs}
\usepackage{bbm}
\usepackage{latexsym}
\usepackage{booktabs,amsfonts,dcolumn}
\usepackage{graphics} 
\usepackage{epsfig} 
\usepackage{amssymb}  %
\usepackage{upgreek}
\usepackage{forest} 
\usepackage{comment}
\usepackage{cancel}
\usepackage{url}
\usepackage{soul}
\usepackage{graphicx,psfrag, color}
\usepackage{latexsym}
\usepackage{bbm}
\usepackage{latexsym}
\usepackage{lettrine}%
\usepackage{dsfont}
\usepackage{array}
\usepackage{amsthm}
\usepackage{multicol}

\usepackage[english]{babel}
\usepackage{graphicx}
\usepackage{pgf, tikz, pgflibraryshapes}
\usepackage{amssymb,bm}
\usepackage{colortbl}
\usepackage{multicol}
\usepackage{stackrel}
\usepackage{pgfplots}
\usepackage{ulem}
\newcommand{\ba}{\begin{array}}
	\newcommand{\ea}{\end{array}}
\usepackage{color}
\usepackage[mathscr]{euscript}

\usepackage{amsfonts}

\newcommand{\integernonnegative}{\mathbb{Z}_{\geq 0}}

\newcommand{\G}{\mathcal{G}} 
\newcommand{\V}{\mathcal{V}} 

\newtheorem{corollary}{Corollary}

\newtheorem{theorem}{Theorem}
\newtheorem{lemma}{Lemma}

\newtheorem{proposition}{Proposition}
\newtheorem{definition}{Definition}
\newtheorem{remark}{Remark}

\theoremstyle{definition}
\newtheorem{example}{Example}

\renewcommand{\eqref}[1]{Eq.~(\ref{#1})}

\definecolor{darkcyan}{rgb}{0.0, 0.55, 0.55}
\definecolor{internationalorange}{rgb}{1.0, 0.31, 0.0}

\usepackage{color}
\usepackage{ulem}

\begin{document}
\title{Ergodic and Metastable Behavior of Stochastic Popularity Dynamics in Social Networks}
\author{Chiara Ravazzi, 
Paolo Frasca, Gaya Cocca
\thanks{
\indent C.~Ravazzi is with the National Research Council of Italy (CNR-IEIIT), c/o Politecnico di Torino, Corso Duca Degli Abruzzi 24, 10129, Torino, Italy. 
			E-mail: chiara.ravazzi@cnr.it.
P.~Frasca is with Univ.\ Grenoble Alpes, CNRS, Inria,  Grenoble INP, GIPSA-lab, F-38000 Grenoble, France. E-mail: paolo.frasca@gipsa-lab.fr.
G.~Cocca was a master's student of Politecnico di Torino, Corso Duca Degli Abruzzi 24, 10129, Torino, Italy. 	E-mail: gaya.coccagc@gmail.com. 
}
\thanks{This work has been supported in part by ANR grant FeedingBias (ANR-22-CE38-0017-01).}}
\date{}

\maketitle

\begin{abstract}
This paper studies popularity dynamics in social networks, where users' engagement with competing channels is driven by intrinsic content quality, social influence, and popularity feedback. We introduce a stochastic model in which users interact through a directed influence network and channel exposure depends jointly on these three mechanisms. We first characterize the resulting dynamics in the absence of intrinsic quality, showing that popularity concentrates asymptotically on a single channel, leading to a winner-takes-all regime. When all channels have positive intrinsic quality, we establish ergodicity of the popularity process and characterize its long-run behavior. We further show that, when intrinsic quality is small, the system exhibits metastable behavior: it can remain for long periods near configurations dominated by a single channel before transitioning to another. These results highlight how weak quality differences can coexist with strong path dependence induced by popularity feedback. We then investigate platform-level interventions and propose controls that compensate for network-induced exposure biases and mitigate excessive concentration. Analytical results are illustrated through numerical simulations and an empirical analysis of YouTube data, whose observed popularity dynamics are consistent with the proposed stochastic framework. The results provide a mathematical framework for understanding the interplay between quality, social influence, and reinforcement in online popularity dynamics.
\end{abstract}



\section{Introduction}
Digital platforms have profoundly transformed the dissemination and consumption of news, information, videos, and cultural products~\cite{10.1080/23808985.2022.2142149}. Content popularity results from the interplay of multiple mechanisms: users are influenced by their social environment \cite{pnas1}, platforms amplify successful content through exposure mechanisms \cite{inproceedings}, and channels produce content that differs in quality or attractiveness \cite{Ciampaglia2018AlgorithmicPopularityBias}. Visibility and engagement are therefore shaped not only by intrinsic content characteristics, but also by feedback effects that can reinforce small initial asymmetries and generate highly uneven outcomes over time \cite{pnas2}.

Popularity dynamics are consequently central to understanding the formation of attention and public opinion. From a methodological perspective, these systems combine interaction, heterogeneity, and randomness, making their long-term behavior difficult to characterize. Empirical studies have documented highly skewed distributions of views, likes, and shares, as well as rich-get-richer effects and winner-take-all phenomena across a variety of platforms \cite{Weng2012CompetitionAmongMemes, LorenzSpreen2019}. Preferential attachment models provide one possible explanation for such phenomena by modeling mechanisms through which already popular items are more likely to attract further attention \cite{Hofstad_2016}. More recently, platform-mediated content exposure has emerged as a central mechanism shaping visibility and engagement \cite{Castaldo2021}, as ranking systems tend to privilege successful content and thereby create endogenous feedback loops \cite{11036812,Rossi2022ClosedLoopOpinionRecommendations,11184854}.

Despite this extensive theoretical literature, several aspects remain insufficiently understood. First, many existing models are either deterministic \cite{11036812} or formulated at an aggregate level \cite{10.1007/978-3-032-13821-7_26}, and therefore do not explicitly account for the randomness of individual engagement. At the microscopic scale, whether a user watches, likes, or interacts with a piece of content is inherently stochastic, even under similar conditions. Such variability may have important dynamical consequences when attention is repeatedly reshaped by social and algorithmic feedback. Second, network influence, popularity reinforcement, and content quality are often studied separately, whereas these mechanisms coexist and interact in practice. A unified framework is therefore needed to disentangle their respective roles. Third, although empirical evidence suggests the coexistence of stable long-term patterns and prolonged episodes of concentration around dominant creators \cite{ext3imp}, the mechanisms through which ergodic behavior and metastable winner-take-all regimes can arise within the same model remain insufficiently understood.

This paper addresses these issues by introducing a stochastic model of popularity dynamics over a social network (Section~\ref{cap:stoc}), featuring users and content providers, which we refer to as channels because of our interest in YouTube as an empirical case study. The model extends the deterministic model of popularity dynamics that we recently introduced in~\cite{11036812} by considering two coupled variables: the attention that each user allocates to each channel and the corresponding individual reaction, such as a view or a like. Attention evolves through three mechanisms: social imitation over the user network, platform-mediated popularity reinforcement, and an exogenous quality term associated with each channel. Individual reactions are then modeled as Bernoulli random variables conditioned on attention. This formulation retains a simple microscopic interpretation while providing sufficient structure to derive analytical results on the long-term behavior of the system.

Our analysis considers two main cases, depending on whether content quality contributes to the dynamics. When quality plays no role, we study monopoly formation and characterize the winning probabilities of competing channels in Section~\ref{sec:no_quality}. We also investigate platform interventions on content exposure, aimed at compensating for structural network advantages and mitigating excessive popularity concentration. Related results on the control of ergodic social dynamics were recently presented in~\cite{11407443}. When quality influences the dynamics, we derive conditions for ergodicity in Section~\ref{sec:all_effects}, and study the small-quality regime and the emergence of metastable behavior in Section~\ref{sec:meta}. We also discuss control strategies for mitigating such metastability.

Numerical experiments are presented in Section~\ref{sec:numerics}, followed by an empirical analysis of YouTube data in Section~\ref{sec:emprical_analysis}. As in related work on opinion dynamics, the purpose of the model is explanatory rather than predictive. We do not aim to reproduce or forecast the exact popularity trajectory of a given channel on a real platform \cite{11627737}, nor to infer platform-specific recommendation rules from data. Rather, we seek to provide a mathematically tractable framework that isolates a few fundamental mechanisms and clarifies how their interaction can generate qualitatively different regimes. The model should therefore be interpreted as a stylized representation of popularity formation, designed to identify mechanisms underlying concentration, persistence, ergodicity, and metastable behavior, some of which are also reflected in the empirical patterns discussed in Section~\ref{sec:emprical_analysis}. Concluding remarks are given in Section~\ref{sec:conclusions}, and all proofs are provided in the Appendix.

\section{Stochastic Modeling of Popularity Dynamics} \label{cap:stoc}
\label{section:model}
We consider a population of $n$ individuals, indexed in the set $\mathcal{V}$, who are users of a social media platform, e.g.\ YouTube or TikTok. At each time $t\in\integernonnegative$,  the users are exposed to different content, proposed by $m$ channels, indexed in set $\mathcal{C} $. The value $X_v^{( i )}(t)\in[0,1]$ represents the level of attention of user $v\in\V$ for content proposed by channel $ i\in\mathcal{C} $ at time $t\in\integernonnegative$. 
We suppose that a user's attention is influenced by the attention given by other users in the social network. The mutual influence is represented using a weighted graph
$\mathcal{G}=(\mathcal{V},\mathcal{E}, P)$, where $\mathcal{E}\subseteq\mathcal{V}\times\mathcal{V}$ describes the influence among the users, and the weighted adjacency matrix $P$ encodes the strength of these connections. The matrix $P$ is supposed to be row stochastic, i.e., with nonnegative entries and $P\mathds{1}_n=\mathds{1}_n$.

We assume that the dynamics of the attention $X_v^{(i)}(t)$ evolves according to the following recursion:
\begin{equation} \label{eq: stoch_model_simpl}
X_v^{(i)}(t+1) = \alpha_v \sum_{w } P_{vw} Y^{(i)}_w(t) + \beta_v \Pi^{(i)}(t) + \gamma_v q^{(i)},
\end{equation}
where
\begin{itemize}
\item $Y_v^{(i)}(t) \sim
\operatorname{Bernoulli}(X_v^{(i)}(t))$
conditionally on $X_v^{(i)}(t)$, representing whether user $v$ engaged with the content proposed by channel $i\in\mathcal{C}$ at time $t\in\integernonnegative$. In particular, $Y_v^{(i)}(t) = 1$ indicates that the user expressed positive engagement (e.g., a like/view), while $Y_v^{(i)}(t) = 0$ corresponds to the absence of engagement.
\item $\Pi^{(i)}(t)$ represents the popularity of channel $i\in\mathcal{C}$ at time $t\in\integernonnegative$ defined as the success of its content normalized against the success of the content of all other channels
\begin{equation} \label{eq:rankingstoch}
\Pi^{(i)}(t) = \frac{\sum_{v \in \mathcal{V}} Y^{(i)}_v(t)}{\sum_{j \in \mathcal{C}} \sum_{w \in \mathcal{V}} Y^{(j)}_w(t)},
\end{equation}
where we adopt the convention that $\Pi^{(i)}(t)={1}/{m  } $ for all $i\in\mathcal{C}$ when $Y^{(i)}_v(t)=0$ for all $v\in\mathcal{V}$ and all $i\in\mathcal{C}$. It should be noticed that $\sum_{i\in\mathcal{C}}\Pi^{(i)}(t)=1$ for all $t\in\integernonnegative$.
\item The parameter $q^{(i)}\in[0,1]$ is an intrinsic feature of channel $i\in\mathcal{C}$, representing the quality of the content proposed by the channel (with the meaning that content of higher quality should receive more attention). 
\item The contributions of social influence, popularity, and quality are weighted through nonnegative coefficients $\alpha_v$, $\beta_v$ and $\gamma_v$, 
such that $
\alpha_v + \beta_v + \gamma_v = 1.
$
\end{itemize}

\subsection{The engagement Markov chain}
The model described in \eqref{eq: stoch_model_simpl} depends on the random variables $\{Y_v^{(i)}(t)\}_{v\in\V,i\in\mathcal{C}}$, which are the only source of randomness. It should be noticed that, given the current state $Y(t)$, the conditional distribution of future
values of $Y$ does not depend on past values, nor on other variables. Therefore, $\{Y(t)\}_{t\in\integernonnegative}$ evolves as a discrete-time Markov chain with $2^{n m  }$ states.
Any state can be represented by a matrix
$
z \in \{0,1\}^{\mathcal{V}\times\mathcal{C}},
$
whose entry $z_{v,i}$ indicates whether user $v$ positively engaged with channel $i$ in that state.

We now define the transition matrix $W$ of the Markov chain. For any two states
$z,z' \in \{0,1\}^{\mathcal{V}\times\mathcal{C}}$, let
$$
W_{z'\leftarrow z} = \mathbb{P}(Y(t+1)=z' \mid Y(t)=z),
\qquad \forall t\in\mathbb{N}_0.
$$
For a specific channel $i\in\mathcal{C}$, we denote by
$$
\mathcal{S}^{(i)}(z') = \mathrm{supp}\left((z')^{(i)}\right)
= \{v\in\mathcal{V} \mid z'_{v,i}=1\}
$$
the subset of users who positively engage with channel $i$ in the state $z'$.
Then the probability $W_{z'\leftarrow z}$ of transitioning from $z$ to $z'$ is
\begin{equation} \label{def:prob}
W_{z'\leftarrow z}
=
\prod_{i\in\mathcal{C}}
\prod_{v\in \mathcal{S}^{(i)}(z')}
\theta_v^{(i)}(z)
\prod_{v\notin \mathcal{S}^{(i)}(z')}
\left(1-\theta_v^{(i)}(z)\right),
\end{equation}
where
\begin{equation}\label{eq:theta}
\theta_v^{(i)}(z)
=
\alpha_v \sum_{w\in\mathcal{V}} P_{vw} z_{w,i}
+ \beta_v \Pi^{(i)}(z)
+ \gamma_v q^{(i)},
\end{equation}
and
$$
\Pi^{(i)}(z)
=\displaystyle
{\sum_{v\in\mathcal{V}} z_{v,i}}\bigg/
{\sum\limits_{j\in\mathcal{C}}\sum\limits_{w\in\mathcal{V}} z_{w,j}},
$$
with the convention that $\Pi^{(i)}(z)=1/m  $ for all $i\in\mathcal{C}$ when
$z_{v,i}=0$ for all $v\in\mathcal{V}$ and all $i\in\mathcal{C}$.

In the following sections, we will study different scenarios, associated with different values of parameters of the model in \eqref{eq: stoch_model_simpl}, and show the corresponding outcomes.

\section{Popularity reinforcement without quality influence}\label{sec:no_quality}
For each channel $i\in \mathcal{C}$, denote 
$$
z^{\star,i}_{v,j}
=
\begin{cases}
1, & j=i,\\
0, & j\neq i,
\end{cases}
\qquad \forall v\in\mathcal V,\ \forall j\in\mathcal C.
$$
Then, we say that the monopoly states are those in 
$$
\mathcal{Z}^{\star}=
\left\{
z^{\star,i}\in\{0,1\}^{\mathcal{V}\times\mathcal{C}}: i\in \mathcal{C}
\right\},
$$
that is, the states corresponding to all attention concentrated on a single channel.

The following theorem 
shows that when users disregard the quality of the content proposed by each channel, popularity concentrates on a single channel.
\begin{theorem}[No quality -- Popularity reinforcement]
\label{thm:no_quality_all_alpha}
Let  $\alpha_v>0,\beta_v > 0$ and $\gamma_v = 0$ for all $v\in\mathcal{V}$. Then, the following facts hold true.
\begin{enumerate}
\item[(i)] 
The Markov chain $Y$ admits absorbing states, and the absorbing states are the monopoly states $\mathcal{Z}^{\star}$.

\item[(ii)]
If $\G$ is strongly connected, then the absorbing state $z^{\star,j}$ is reachable from any initial condition $Y(0)$ satisfying
$
\sum_{v\in\mathcal{V}} Y_v^{(j)}(0) \geq 1.
$

\item[(iii)]
If $\G$ is strongly connected, then the Markov chain converges to one of the absorbing states in $\mathcal{Z}^{\star}$.
\end{enumerate}
\end{theorem}

Notice that, if instead $\beta_v=0$ $\forall v\in\V$, channels evolve independently: each channel can be chosen or abandoned by all users, and therefore multiple channels can survive simultaneously. On the contrary, if $\beta_v>0$, however, the popularity term $\Pi^{(i)}(t)$ introduces global competition between channels and makes only monopoly states absorbing. %

The following example illustrates the asymptotic behavior.
\begin{example} \label{ex:stoch_gamma0}
We consider a social network with $n=500$ users following $m=5$ channels. 
We sample the parameters $\alpha_v$, $v\in\mathcal{V}$, and the initial conditions $X_v^{(i)}(0)$, $v\in\mathcal{V}$, $i\in\mathcal{C}$, independently from uniform distributions on $[0,1]$. The graph is a random graph with self-loops, edge probability $equal to 0.5$, and uniformly generated edge weights.
We show in Figure~\ref{fig:gamma0} attention of the users (left) and the evolution of popularity (right) as a function of time, iterating the updates for 400 time instants.
\end{example}

\begin{figure}
        \includegraphics[width=0.99\columnwidth]{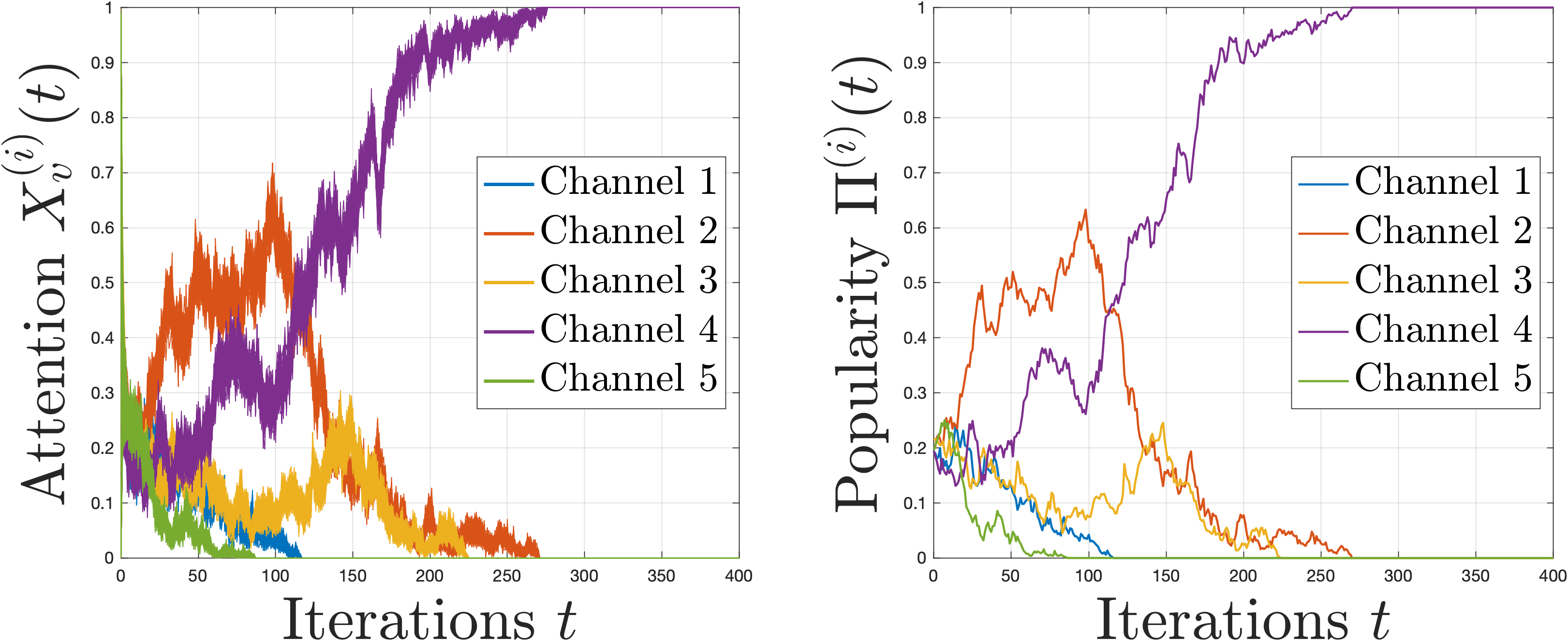}
        \caption{Evolution of the attention of the users (left) and popularity (right) in Example~\ref{ex:stoch_gamma0}.}
        \label{fig:gamma0}
\end{figure}

Theorem~\ref{thm:no_quality_all_alpha} guarantees that, under strong
connectivity, the stochastic process reaches a monopoly
configuration almost surely. However, the theorem does not identify which channel
will gain the support of the network. 

In order to investigate this question, let us now denote the first hitting time of the set of monopoly states by
\begin{equation}
\tau_{\mathcal Z^\star}
:=
\inf\left\{
t\geq 0:
Y(t)\in\mathcal Z^\star
\right\}
\label{eq:absorption_time}
\end{equation}
and define, for every
channel $i\in\mathcal C$, every initial state
$z\in\{0,1\}^{\mathcal V\times\mathcal C}$, and every influence matrix
$P$, the probability that the process is absorbed in the monopoly state associated with channel $i$ starting from initial condition $z$
\begin{equation}
h_i(z;P)
:=
\mathbb P_z
\left(
Y(\tau_{\mathcal Z^\star})=z^{\star,i}
\right).
\label{eq:winner_probability}
\end{equation}
The following result gives a decomposition of this winning probability
into a network contribution and a cumulative platform contribution. For its statement, we need some definitions.

Define $A=\operatorname{diag}(\alpha_v)$ and observe that, under the assumptions of
Theorem~\ref{thm:no_quality_all_alpha}, the nonnegative matrix
$AP$
is substochastic. 
Let
$\lambda
:=
\rho\!\left(AP
\right)$ and notice that $0<\lambda<1$, and let $\varpi\in\mathbb R_{>0}^{\mathcal V}$ be its normalized
left Perron eigenvector, i.e.,
$\varpi^\top
AP
=
\lambda\varpi^\top,
\varpi^\top\mathds{1}=1.$

Furthermore, for every channel $i\in\mathcal C$ and every state
$z\in\{0,1\}^{\mathcal V\times\mathcal C}$, define
$C_i(z):=\varpi^\top z^{(i)}
$, representing the support for channel~$i$ weighted
according to the influence network. It should be noted that, since
$\varpi^\top\mathds{1}=1$, we have
$C_i(z^{\star,j})=\mathbbm{1}_{\{i=j\}}.
$

\begin{theorem}[Decomposition of the winning probability]
\label{thm:winner_decomposition}
Under the assumptions of
Theorem~\ref{thm:no_quality_all_alpha}, for every initial state
$z\in\{0,1\}^{\mathcal V\times\mathcal C}$ and every
$i\in\mathcal C$,
\begin{equation}
\begin{split}\small
&h_i(z;P)
=
C_i(z)
\\
&\quad+
(1-\lambda)
\mathbb E_z
\left[
\sum_{t=0}^{\tau_{\mathcal Z^\star}-1}
\left(
\Pi^{(i)}(Y(t))
-
C_i(Y(t))
\right)
\right].
\end{split}
\label{eq:winner_decomposition}
\end{equation}
\end{theorem}

Theorem~\ref{thm:winner_decomposition} distinguishes two contributions:
The first one is the initial support measured according to the
influence network; the second one represents the cumulative correction
generated by the platform popularity term during the transient
evolution before the process reaches a monopoly state.
In the next section, we use a two-community setting to illustrate how the same amount of support can lead to different monopoly probabilities depending on where that support is located in the network, and how platform reinforcement can modify or even reverse this structural advantage.

\subsection{Winning probability in a two-community setting}
We now consider a two-community directed influence network, for which both terms in Theorem~\ref{thm:winner_decomposition} can be characterized explicitly. Let $\mathcal V=\mathcal V_1\cup\mathcal V_2$, with $|\mathcal V_1|=n_1$, $|\mathcal V_2|=n_2$, and $n=n_1+n_2$. Nodes in the same community have the same aggregate influence profile. Let
\begin{equation}
B=
\begin{pmatrix}
1-a & a\\
b & 1-b
\end{pmatrix},
\qquad a,b\in(0,1),
\label{eq:core_periphery_B}
\end{equation}
and define $P_{vw}=B_{rs}/n_s$ for $v\in\mathcal V_r$ and $w\in\mathcal V_s$ with $s\in\{1,2\}$. When $a\neq b$, influence between the two communities is asymmetric.
For a state $z$, let $k_{r,i}(z):=\sum_{v\in\mathcal V_r}z_{v,i}$ denote the support for channel $i$ in community $r$, and let $k_i(z):=k_{1,i}(z)+k_{2,i}(z)$.

\begin{proposition}[Two-community directed network]
\label{thm:two_community_network}
Consider the network defined above and homogeneous parameters $\alpha_v=1-\beta$ and $\beta_v=\beta$ for all $v\in\mathcal V$, with $\beta\in(0,1)$. Then:
\begin{enumerate}

\item[(i)] For every state $z$ and channel $i$,
\begin{equation}
C_i(z)=
\frac{b}{a+b}\frac{k_{1,i}(z)}{n_1}
+
\frac{a}{a+b}\frac{k_{2,i}(z)}{n_2}.
\label{eq:core_periphery_network_score}
\end{equation}

\item[(ii)] Define $K_{r,i}(t):=\sum_{v\in\mathcal V_r}Y_v^{(i)}(t)$ and $K(t):=(K_{r,i}(t):r=1,2,\ i\in\mathcal C)$. Then, $K(t)$ is a Markov chain on the aggregate state $k=(k_{r,i})$. Let $k_i:=k_{1,i}+k_{2,i}$, $S(k):=\sum_{j\in\mathcal C}k_j$, and
$$
p_i(k):=
\begin{cases}
k_i/S(k), & S(k)>0,\\
1/m  , & S(k)=0.
\end{cases}
$$
Conditionally on $K(t)=k$,
$
K_{r,i}(t+1)\sim
\operatorname{Binomial}\left(n_r,\theta_{r,i}(k)\right),
$
where
\begin{equation}
\theta_{r,i}(k)
=(1-\beta)\sum_{s=1}^2 B_{rs}\frac{k_{s,i}}{n_s}
+\beta p_i(k).
\label{eq:core_periphery_theta}
\end{equation}

Let $\widehat Q$ be the restriction of the transition matrix of the aggregate chain to its transient states and define

$$
r_i(k):=
p_i(k)
-
\frac{b}{a+b}\frac{k_{1,i}}{n_1}
-
\frac{a}{a+b}\frac{k_{2,i}}{n_2}.
$$

Then, with $k(z):=(k_{r,i}(z):r=1,2,\ i\in\mathcal C)$,
\begin{equation}
h_i(z;P)
=C_i(z)
+
\beta
\left[(I-\widehat Q)^{-1}r_i\right](k(z))
\label{eq:core_periphery_closed_hitting}
\end{equation}

\end{enumerate}
\end{proposition}

Proposition~\ref{thm:two_community_network} makes explicit both components of Theorem~\ref{thm:winner_decomposition}. The first depends on the distribution of initial support across the two communities and on the network asymmetry through $bn_2-an_1$. The second is the cumulative platform correction and can be computed through the fundamental matrix of the reduced community-level Markov chain.

\subsection{Platform control} 
A platform may have an interest in actively shaping content exposure in order to mitigate undesirable concentration effects and preserve a sufficiently diverse semantic ecosystem. A natural control action consists in modifying the platform exposure term by introducing a control input $u(t)\in\mathbb{R}^{\mathcal C}$, namely
$
\Pi_u^{(i)}(t)=\Pi^{(i)}(t)+u_i(t),
 i\in\mathcal C,
$
with
$
\sum_{i\in\mathcal C}u_i(t)=0
$
and $\Pi_u^{(i)}(t)\in[0,1]$ for every $i\in\mathcal C$.
The controlled attention dynamics then become
\begin{center}$
X_v^{(i)}(t+1)
=
\alpha_v\sum_{w\in\mathcal V}P_{vw}Y_w^{(i)}(t)
+
\beta_v \Pi_u^{(i)}(t)
+
\gamma_v q^{(i)}.
$\end{center}
This formulation allows the platform to redistribute exposure across channels while preserving the total amount of exposure
and the control input can be designed according to different objectives, such as compensating for structural advantages induced by the influence network or mitigating excessive popularity concentration.
Specific examples of control laws are introduced below in the two-community setting, for which \eqref{eq:core_periphery_theta} implies that controlled aggregate dynamics is
$$
\theta_{r,i}^{u}(k)
=
(1-\beta)\sum_{s=1}^2 B_{rs}\frac{k_{s,i}}{n_s}
+
\beta p_i^{u}(k).
$$
\begin{example}[Network-bias compensation control]
For network-bias compensation, one may choose
\begin{center}$
u_i(k)=-\eta\,\Delta_i(k),
\qquad
\Delta_i(k):=
\frac{bn_2-an_1}{n(a+b)}
\left(
\frac{k_{1,i}}{n_1}
-
\frac{k_{2,i}}{n_2}
\right),
$\end{center}
so as to reduce the advantage of channels favored by the network structure.
The resulting exposure vector is projected onto the probability simplex
to preserve nonnegativity and unit total mass.

For a numerical example, we consider a two-community network with $n=500$ users and $m=5$ channels. 
The model parameters are fixed to
$
\alpha=0.6,$ $\beta=0.4,$ 
while the inter-community interaction matrix is defined by
$
a=0.08,$ $b=0.42.$
We define the ratio
$
\rho={n_1}/{n}
$
and we let it vary in the interval $[0.2,0.9]$, while the control gain $\eta$
ranges in $[0,1]$. For each pair $(\rho,\eta)$, the hitting
probabilities are estimated over $2000$ Monte Carlo realizations, with
a maximum simulation horizon of $5000$ time steps. The initial support
profiles are kept fixed in relative terms across the two communities, so
that Channels~1 and~2 are initially favored by opposite communities.

Figure~\ref{fig:hitting_control} illustrates how the control gain $\eta$ affects the relative hitting probabilities as the ratio
$
\rho={n_1}/{n}
$
varies. The left panel reports the difference $h_1-h_2$
and the solid black contour identifies the locus $h_1= h_2$. 
As $\rho$ increases, progressively stronger control is required to compensate for the structural advantage induced by the imbalance in community sizes.
This behavior is summarized in the right panel, which reports the minimum control gain $\eta^\star(\rho)$ required to reverse the ranking. For relatively balanced communities, the required intervention is small, while $\eta^\star(\rho)$ increases rapidly as the first community becomes larger. Beyond a critical range of $\rho$, no ranking reversal can be achieved within the admissible control range $\eta\leq 1$, showing that the effect of network structure eventually dominates the available platform intervention. The dashed vertical line marks the structural reference value
$
\rho^\star={b}/({a+b}),
$
which provides a natural benchmark for the imbalance induced by the underlying transition rates.
\begin{figure}
        \includegraphics[width=0.99\columnwidth]{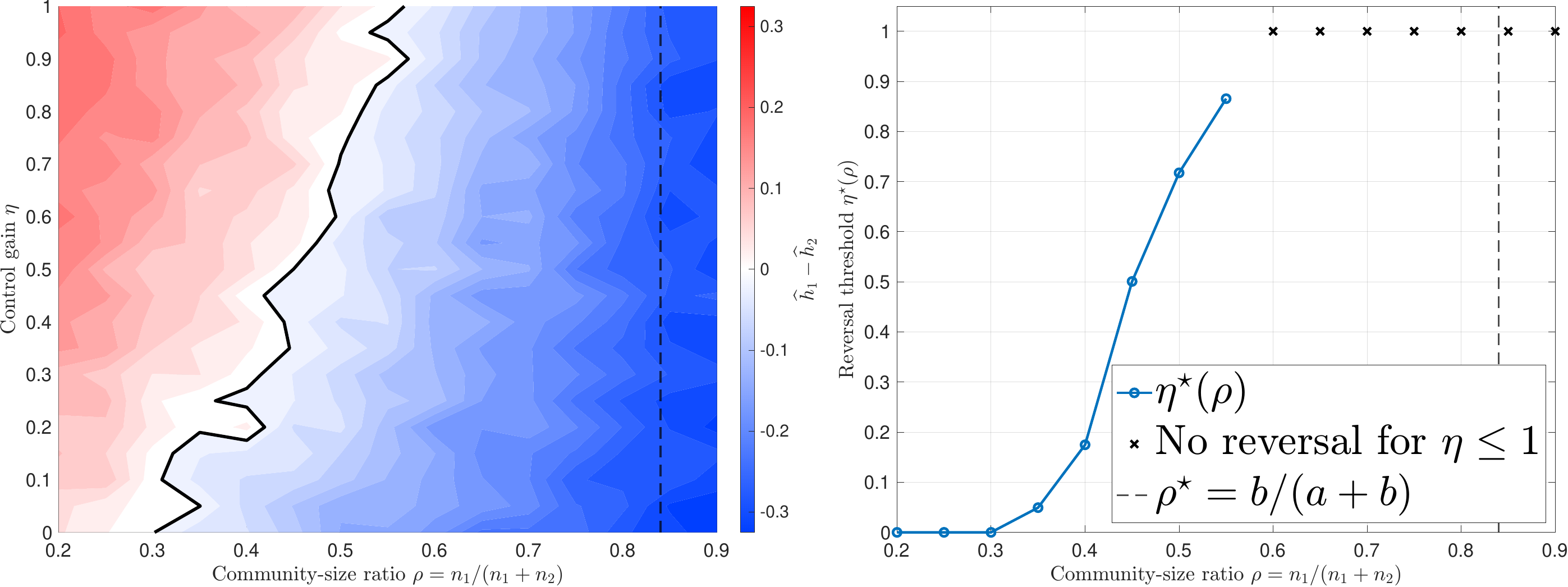}
        \caption{Network bias compensation control: Hitting-probability difference (left) and minimum control for ranking reversal (right)}
        \label{fig:hitting_control}
\end{figure}
\end{example}

\begin{example}[Anti-concentration control] The platform may use an anti-concentration control that decreases the exposure of channels whose popularity exceeds a reference level and redistributes it among the others, for instance
\begin{center}$
u_i(k)=-\eta\left(p_i(k)-\frac{1}{m  }\right).
$\end{center}
\begin{figure}
        \includegraphics[width=0.99\columnwidth]{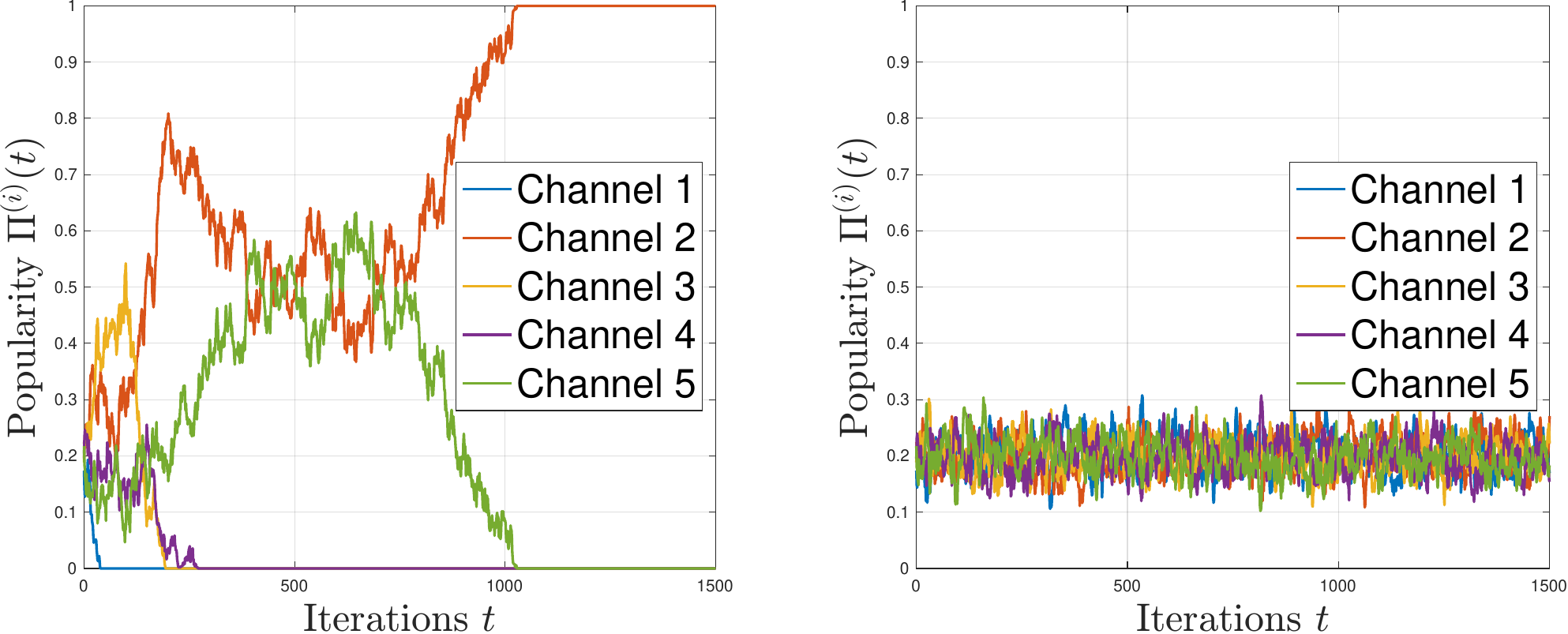}
\caption{Anti-concentration control: comparison of uncontrolled (left) and controlled (right) popularity dynamics.}
        \label{fig:anticoncentration_control}
\end{figure}
Figure \ref{fig:anticoncentration_control} compares the uncontrolled and controlled popularity dynamics for
$n=500$ users and $m=5$ channels over $T=1500$ iterations, with
$\alpha=0.10$, $\beta=0.90$, $\gamma=0$, and anti-concentration gain
$\eta=0.15$. The interaction network is generated with edge probability
$p_{\mathrm{edge}}=0.08$, while the initial engagement probability is
$\theta_0=0.05$. Without control, one channel eventually dominates, producing
a winner-takes-all outcome. In contrast, the anti-concentration control keeps
the five channel popularities fluctuating around the balanced value
$1/m=0.2$.
\end{example}

\section{Ergodic behavior with quality influence}\label{sec:all_effects}
We now focus on the general dynamics, in which all terms appear, i.e. $\gamma_v>0$ for all $v\in\mathcal{V}$.

\begin{theorem}[Ergodicity]\label{thm:ergodicity}
Let $\gamma_v>0$ for all $v\in\mathcal{V}$ and
$q^{(i)}\in(0,1)$ for all $i\in\mathcal C$. Consider the Markov chain
$\{Y(t)\}_{t\in\integernonnegative}$ with transition matrix~$W$
defined in~\eqref{def:prob}. Then, the following facts hold true.
\begin{enumerate}
\item $Y$ admits a unique invariant distribution $\nu$ satisfying $\nu^{\top}W=\nu;$
\item $Y$ converges in distribution to a random variable $Y_{\infty}\sim\nu$; 
\item $Y$ is ergodic in the time-average sense, i.e., 
almost surely
\begin{equation}\label{eq:ergodicity}
\mathbb{E}_\nu[f(Y_{\infty})] = \lim_{T \to \infty} \frac{1}{T} \sum_{t=0}^{T-1} f(Y(t))
\end{equation}
for any function $ f $ that is measurable and integrable with respect to the invariant measure $ \nu $.
\end{enumerate}\end{theorem}

\begin{corollary}\label{corol:X}
If $\gamma_v>0$ for all $v\in\mathcal{V}$ and $q^{(i)}\in(0,1)$ for all $i \in \mathcal{C}$, then
\begin{enumerate}
\item The processes $\{X(t)\}_{t\in\integernonnegative},\{\Pi(t)\}_{t\in\integernonnegative}$ respectively converge in distribution to random variables $X_{\infty}$, $\Pi_\infty$. 
\item The processes $\{X(t)\}_{t\in\integernonnegative},\{\Pi(t)\}_{t\in\integernonnegative}$ are ergodic.
\end{enumerate}
\end{corollary}

\begin{example}\label{ex3}
We consider a social network with $500$ users following five channels.
The user-interaction graph is generated according to the same construction
as in Example~\ref{ex:stoch_gamma0}. For each user $v\in\mathcal{V}$, the parameters
$\alpha_v$, $\beta_v$, and $\gamma_v$ are sampled independently
from a uniform distribution on $[0,1]$ and subsequently normalized so that
$
\alpha_v+\beta_v+\gamma_v=1.
$
The channel quality factors are fixed to five equally spaced values in the interval $[0.1,0.9]$ and the initial conditions $X_v^{(i)}(0)$,
$v\in\mathcal{V}$, $i\in\mathcal{C}$,
are sampled independently from a uniform distribution on $[0,1]$. We show in Figure~\ref{fig:all_terms_new} the resulting evolution of the popularity and its time average.
\end{example}
\begin{figure}
        \includegraphics[width=0.99\columnwidth]{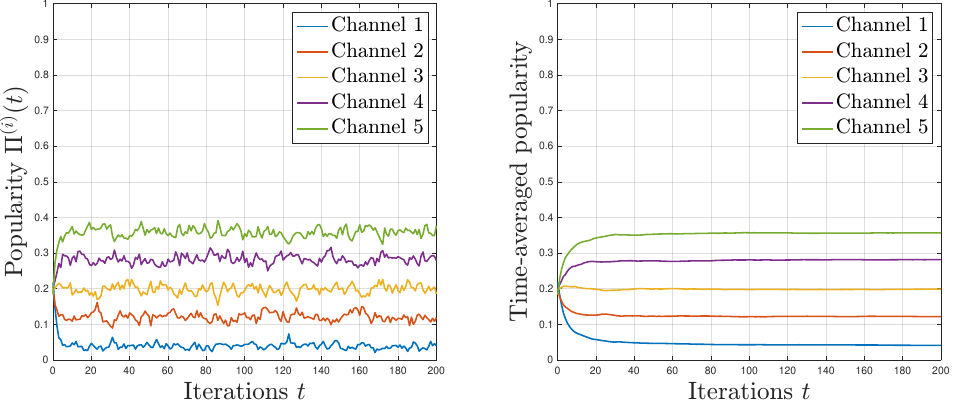}
        \caption{Evolution of popularity (left)  and its time averages (right) in Example~\ref{ex3}.}
        \label{fig:all_terms_new}
\end{figure}

Let $\nu$ be the invariant measure of $W$. Then, for any $\varepsilon\in(0,1)$, the mixing time is defined as follows
$$
\tau_{\mathrm{mix}}(\epsilon)
:= \min\left\{ t\ge0:
\max_{x\in\{0,1\}^{\mathcal{V}\times \mathcal{C}}} \| W^t(x,\cdot)-\nu\|_{\mathrm{TV}} \leq \varepsilon
\right\}.
$$

\begin{proposition}[Mixing-time]\label{prop:mix_time}
The mixing time of the Markov chain $Y$ is such that 
\begin{equation}\label{eq:mix_time}\tau_{\mathrm{mix}}(\varepsilon)
\leq\left\lceil \frac{1}{(\gamma_{\min})^{n m  }}\log\frac{1}{\varepsilon}\right\rceil.
\end{equation}
\end{proposition}
This result 
 yields an explicit upper bound on the mixing time. In particular, for a fixed number of users and channels, the bound decreases as $\gamma_{\min}$ increases. On the other hand, when $\gamma $ is small, Proposition~\ref{prop:mix_time} is less informative. We propose in the next section a specific analysis for $\gamma$ close to~0.

\section{Metastable behavior in the small-$\gamma$ regime}\label{sec:meta}

In the following analysis, we show that, when $\gamma$ is close to~$0$, the dynamics exhibit a metastable behavior. In particular, the exit probability from a monopoly state is $O(\gamma)$ and the probability of being in a monopoly state is close to one. This means that the states $Y(t)$ and vector $\Pi(t)$ spend long periods in a monopoly state before transitioning from one monopoly state to another.

\subsection{Exit time from monopoly states}

It is worth reminding that if $\gamma_v>0$ for all $v\in\mathcal V$, then ergodicity implies that every non-empty subset of the state space is eventually visited with probability one. We are now interested in understanding how much time is spent in a monopoly state.

\begin{theorem}[Exit probability from a monopoly state]\label{thm:exit}
Let $q^{(i)} \in [q_{\min},q_{\max}] \subset (0,1)$
for all $i\in\mathcal C$. Assume
$0<\gamma_{\min}\leq\gamma_v\leq\gamma_{\max}$
for all $v\in\mathcal V$.
For sufficiently small $\gamma_{\max}$, every $i\in\mathcal C$
satisfies
$$
c_1^{(i)}\gamma_{\min}
\le
\mathbb P\left(Y(t+1)\neq z^{\star,i}\mid Y(t)=z^{\star,i}\right)
\le
c_2^{(i)}\gamma_{\max},
$$
with  
$c_1^{(i)}=n \big(1-q^{(i)}\big)/2$, 
$
c_2^{(i)}=n \big((1-q^{(i)})+(m  -1)q_{\max}\big).
$
\end{theorem}

Theorem~\ref{thm:exit} guarantees that the exit probability from an exact monopoly state is uniformly bounded by quantities proportional to $\gamma$, provided $\gamma$ is sufficiently small. We now estimate the corresponding exit time.

\begin{definition}[Exit time from $z^{\star,i}$]
For each $i\in\mathcal C$, assuming
$Y(0)=z^{\star,i}$, define
$
\tau_i
:=
\inf\{t\geq 1:\ Y(t)\neq z^{\star,i}\}.
$
\end{definition}

\begin{corollary}[Expected time in a monopoly state]\label{prop:geom}
Under the hypotheses of Theorem~\ref{thm:exit}, for every $i\in\mathcal C$ there exist constants $\chi_1^{(i)},\chi_2^{(i)}>0$ such that
$
{\chi_1^{(i)}}/{\gamma_{\max}}
\le
\mathbb E[\tau_i]
\le
{\chi_2^{(i)}}/{\gamma_{\min}}.
$
\end{corollary}

\subsection{Concentration on the set of monopoly states}

We now show that, when $\gamma$ is small, the invariant distribution concentrates on the monopoly states.
\begin{theorem}[Concentration on $\mathcal{Z}^{\star}$]\label{thm:metastability}
Let $\nu_\gamma$ be the invariant distribution of the Markov chain $Y$ and denote 
$
\nu_\gamma(\mathcal{Z}^{\star}):=\sum_{y\in\mathcal{Z}^{\star}}\nu_\gamma(y).
$
Assume that the parameters $\alpha_v$, $\beta_v$, $\gamma_v$ are all positive. Then, 
$ \lim_{\gamma\to 0}\nu_\gamma(\mathcal{Z}^{\star})=1.$
\end{theorem}
Being the Markov chain ergodic, this fact implies that also the fraction of time that the dynamics spends on the monopoly states goes to one as $\gamma$ goes to zero.
We refer to this behavior as the Markov chain locking-in on the monopoly states.

\begin{remark}[Mitigating metastable lock-in through platform exploration]
The metastable behavior described above suggests a natural platform-level intervention based on controlled exploration. The popularity-based exposure term $\Pi(z)$ can be replaced by
$
\mu_\eta(z)=(1-\eta)\Pi(z)+\eta d,
\  \eta\in(0,1],
$
where $d$ is in the interior of the simplex, thereby ensuring that every channel receives a strictly positive exposure. 

This modification has two immediate consequences. First, if $\beta_v>0$ for all $v\in\mathcal V$, positive exploration removes the absorbing monopoly configurations even when $\gamma_v=0$ for all~$v$ and renders the resulting Markov chain irreducible and aperiodic. Hence, platform exploration can play the same qualitative regularizing role as the intrinsic-quality term in preventing permanent lock-in.
Second, 
the probability of leaving a monopoly configuration acquires an additional contribution proportional to the exploration intensity. More precisely, for suitable positive constants depending on the model parameters,
\begin{center}$
\mathbb P_\eta
\left(
Y(t+1)\notin\mathcal Z^\star
\,\middle|\,
Y(t)\in\mathcal Z^\star
\right)
=
\Omega(\gamma+\eta),
$\end{center}
which suggests the residence-time scaling
$
\mathbb E[\tau_{\mathcal Z^\star}]
=
O\!\left((\gamma+\eta)^{-1}\right).
$
Thus, intrinsic quality and platform exploration act as two complementary mechanisms against popularity lock-in. 
\end{remark}

\section{Numerical Experiments}\label{sec:numerics}

In this section, we present some numerical results to characterize the long-term behavior of popularity $\Pi^{(i)}(t)$ and highlight how intrinsic quality, social influence, and platform feedback are reflected in the empirical statistics.

\subsection{Empirical distribution of popularity}

\begin{figure*}
\begin{center}
\begin{tabular}{ccc}
\includegraphics[width=0.3\columnwidth]{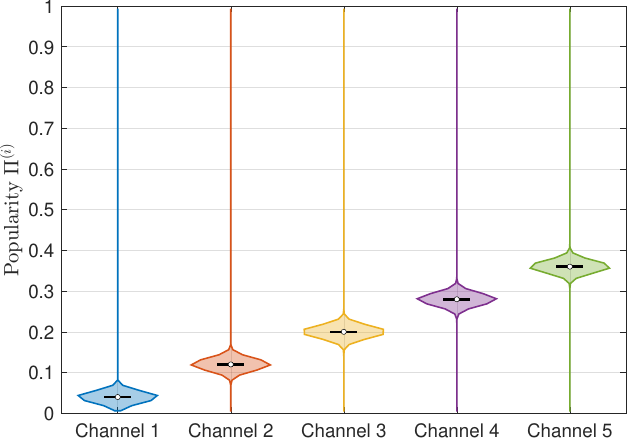}&
\includegraphics[width=0.3\columnwidth]{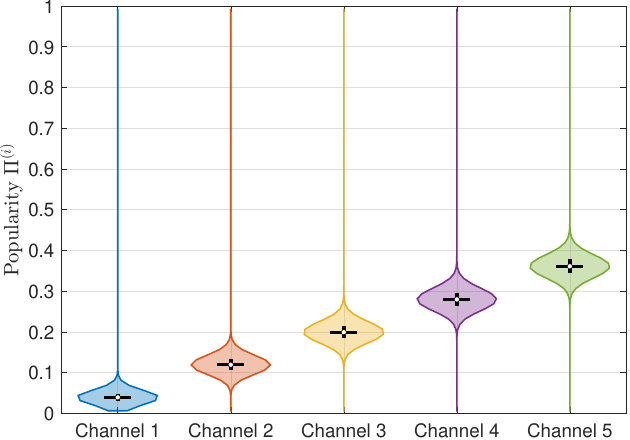}&
\includegraphics[width=0.3\columnwidth]{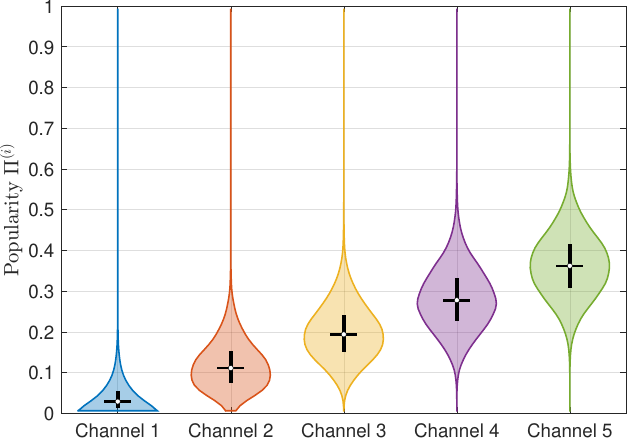}\\
(a)&(b)&(c)
\end{tabular}
\caption{Violin plots of the long-run popularity $\Pi^{(i)}(t)$ for each channel $i=1,\dots,5$, estimated from $60$ Monte Carlo runs: (a) $\alpha=0.1$, $\beta=0.1$, and $\gamma=0.8$; (b) $\alpha=0.45$, $\beta=0.45$, and $\gamma=0.1$; (c) $\alpha=0.09$, $\beta=0.9$, and $\gamma=0.01$.}\label{fig:empirical_pdf}\end{center}
\end{figure*}

We consider a social network of $n=500$ users and $m=5$ channels, with trajectories simulated for $T=5000$ discrete time steps. The initial condition is generated by independently sampling Bernoulli variables with parameter $\theta_0=0.01$. To reduce the impact of the initial transient, all statistics are computed over the last $L=2500$ steps, i.e., $t\in{T-L+1,\dots,T}$. Results are averaged over $60$ independent Monte Carlo runs.
Social influence is modeled by a directed Erd\H{o}s--R\'enyi graph with connection probability $p=0.08$. The influence matrix $P$ is obtained from the adjacency matrix by row-wise normalization. Model weights are homogeneous, with $\alpha_v=\alpha_0$, $\beta_v=\beta_0$, and $\gamma_v=\gamma_0$ for all $v\in\mathcal V$. Channel qualities $q^{(i)}$ are equally spaced in $[0.1,0.9]$, imposing a clear quality ordering.

Figure~\ref{fig:empirical_pdf} shows the empirical long-run distributions of popularity $\Pi^{(i)}(t)$ using violin plots. Figure~\ref{fig:empirical_pdf}(a) considers $\alpha_0=\beta_0=0.1$ and $\gamma_0=1-\alpha_0-\beta_0=0.8$, so that quality is the dominant contribution. The distributions are consequently ordered according to $q^{(i)}$: as quality increases, the corresponding violin shifts toward higher popularity values, whereas low-quality channels concentrate their probability mass at small popularity values, suggesting that they rarely obtain a significant share of engagement. The stochasticity induced by Bernoulli engagement and random network realizations is reflected in the vertical spread of the violin plots, while the separation of their central regions indicates a predominantly quality-driven regime.
The next two configurations illustrate how this picture changes as intrinsic quality is reduced and popularity feedback becomes increasingly dominant. With $\alpha_0=\beta_0=0.45$ and $\gamma_0=0.10$, the regime changes radically: users' attention is driven primarily by social contagion and popularity feedback, while intrinsic quality becomes secondary. Accordingly, in Figure~\ref{fig:empirical_pdf}(b), the popularity distributions become broader, indicating greater variability around their typical values. Channels with different qualities can exhibit similar distributions, and, in some realizations, a lower-quality channel can attain high popularity through initial advantages and algorithmic amplification.
Furthermore, larger values of $\beta_0$ make the process highly self-reinforcing, as illustrated in Figure~\ref{fig:empirical_pdf}(c). When a channel gains even a moderate engagement advantage, the popularity-feedback component increases its probability of being chosen/engaged at the next step. This produces substantially broader violin plots and greater vertical spread.

\subsection{Numerical investigation of metastability}
In this section we consider a social network with $n=20$ users and $m=3$ channels. 
In our setup, the influence network is sampled once from a directed Erdős- Rényi graph with connection probability $p=0.5$, self-loops, and independent uniform edge weights, and the resulting influence matrix is row-normalized. 
The quality parameters are set to $q^{(1)}=0.1, q^{(2)}=0.4, q^{(3)}=0.5$ and the initial individual preferences $X_v^{(i)}(0)$ are generated uniformly at random in the range $[0,1]$. 
We then consider three main regimes: weak metastability ($\gamma_v=0.1$), moderate metastability ($\gamma_v=0.01$),
and strong metastability ($\gamma_v=0.001$). In all cases, the $\alpha_v$ are drawn from a uniform distribution between 0 and 0.5, and the platform influence is given by $\beta_v=1-\alpha_v-\gamma_v$.
To estimate the stationary expected popularity profile, we additionally run $200$ independent Monte Carlo trajectories of length $1.5\cdot10^4$ and average the popularity over the final $5\cdot10^3$ samples of each run.

\begin{figure*}
\begin{center}
        \includegraphics[width=0.97\columnwidth]{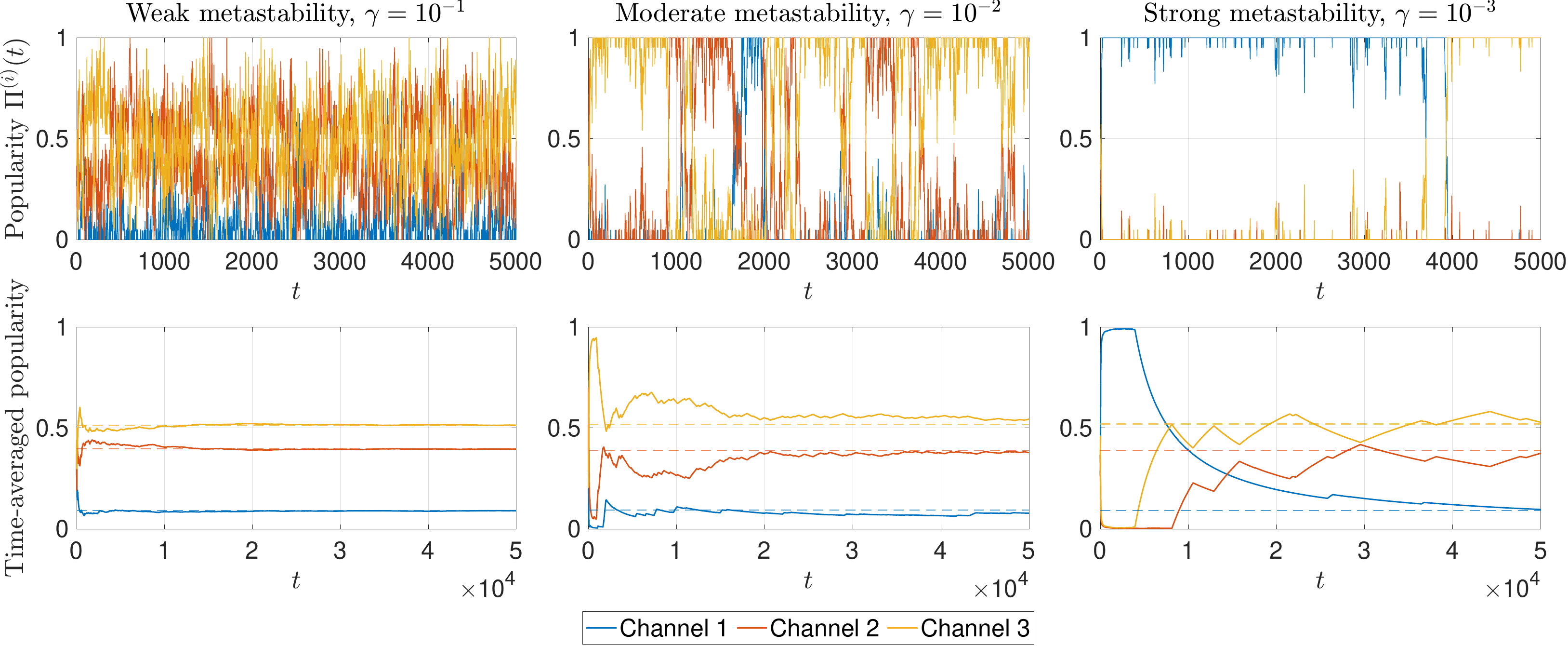}\\
  \end{center}  \caption{Popularity (top) and time-averaged popularity (bottom) as functions of time in different regimes of quality influence.} 
    \label{fig:cp}
\end{figure*}
\begin{figure*}[ht]
        \begin{tabular}{ccc}
            \includegraphics[width=0.3\columnwidth]{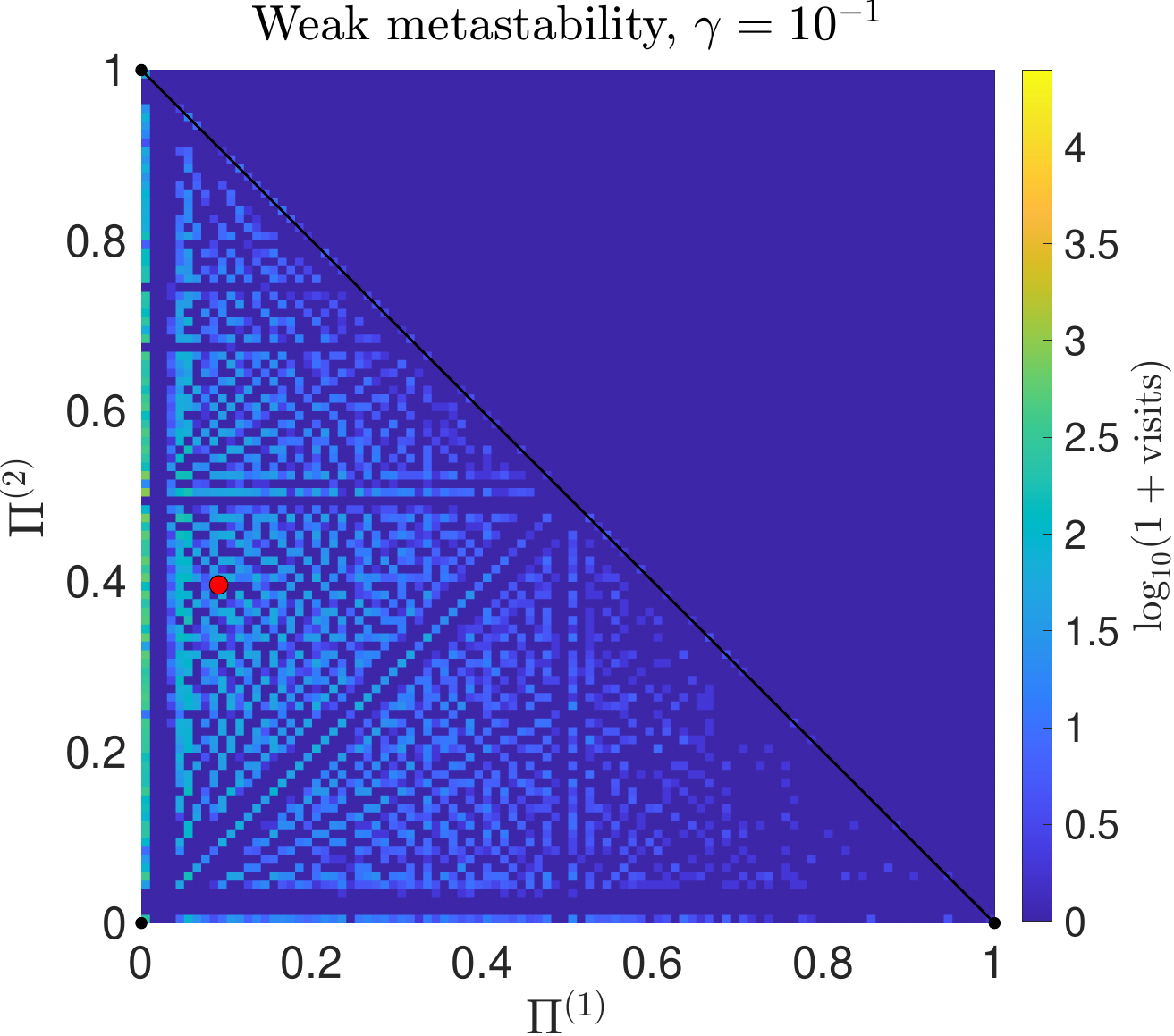}&
                \includegraphics[width=0.3\columnwidth]{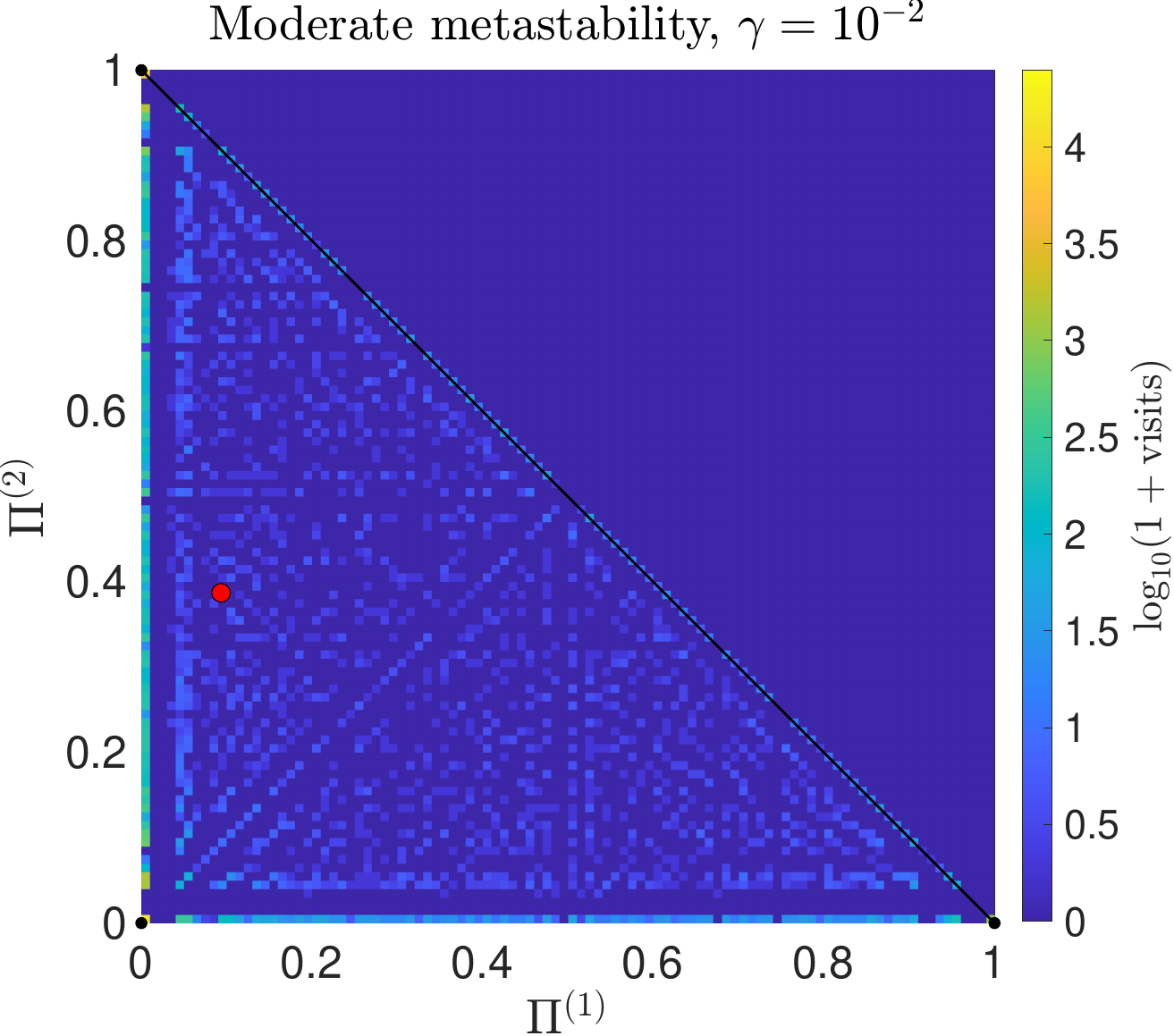}
&    \includegraphics[width=0.3\columnwidth]{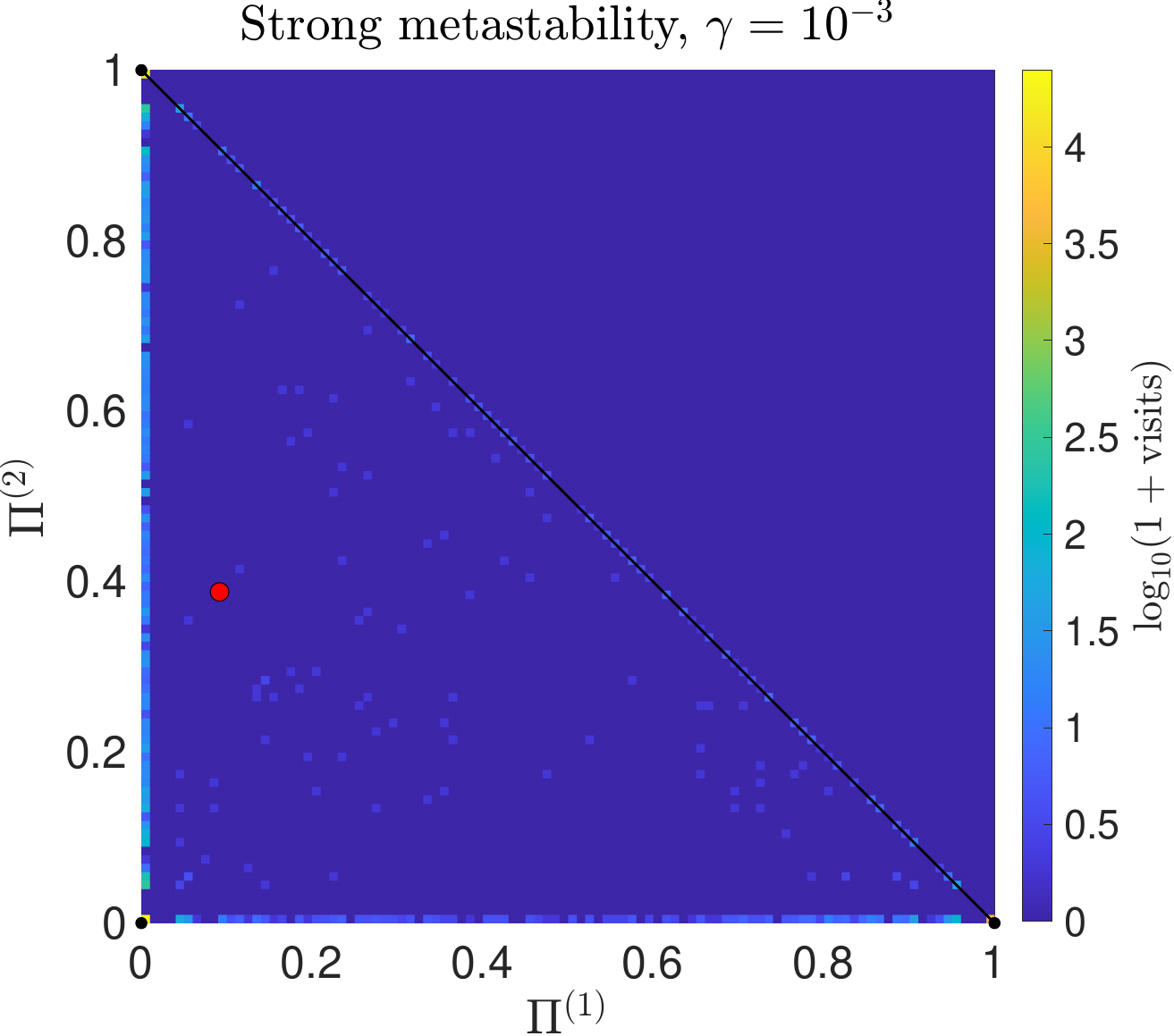}
\end{tabular}
\caption{Popularity heat maps under different quality-influence regimes; the red dot marks the limiting time-averaged popularity.}    \label{fig:cp_heat}
\end{figure*}

Figure~\ref{fig:cp} shows the evolution of the popularity (top plot) and of its time average (bottom plot) for each regime.
As already observed, the system is ergodic and the time averages approach a stable profile. The dashed lines correspond to the estimated $\mathbb{E}_{\nu}[\Pi_{\infty}]$. 
When $\gamma$ approaches zero, then $\Pi^{(i)}(t)$ spends more time in states in which popularity is concentrated on a single channel, remaining in each monopoly state for longer and having less pronounced oscillations.

The heatmap in Figure~\ref{fig:cp_heat} provides further insight into the states visited by the popularity process. 
The axes correspond to the popularity of Channels~1 and~2, and the points $(1,0)$, $(0,1)$, and $(0,0)$ correspond, respectively, to popularity concentrated on Channels~1, 2, and~3.
The red dot indicates the estimated stationary expected popularity profile, while color intensity indicates visitation frequency. Under weak metastability, the most frequently visited state corresponds to popularity concentrated on Channel~2, followed by Channel~1: as expected, popularity tends to concentrate on the channel with higher quality.
As $\gamma$ decreases, then the empirical visitation measure shifts toward the boundary of the popularity simplex, indicating increasingly persistent highly concentrated or quasi-monopoly configurations.

\section{Empirical evidence on YouTube}\label{sec:emprical_analysis}
In this section, we provide empirical evidence of stochastic ergodic behaviors in popularity data from YouTube, thereby supporting the development of a stochastic model of popularity like ours.

The raw dataset for our study is a collection of 316908 videos, each characterized by: channel of publication, identifying name, video title, date and hour of publication, number of views, number of likes, number of comments, and duration.
These videos constitute the whole YouTube publication history, from 2006 to 2023, of the 47 most prominent French news channels (identified with the help of French media experts). Within this initial dataset, we restrict our attention to the period June 2020--July 2023 (after the COVID-19 disruptions) and to the 21 channels that have the most extensive and uninterrupted publication practices.

We choose to measure popularity as the number of views of the videos proposed by a channel, suitably normalized. More precisely, if $N^{(i)}(t)$ represents the number of views for channel $i$ at time $t$, the normalized popularity measure is
\begin{equation} \label{def:visualz_norm}
\Pi ^{(i)}(t) = \frac{ N^{(i)}(t)}{\sum_{j\in\mathcal{C}}N^{(j)}(t)}.
\end{equation}
The time step used to evaluate the trends in popularity is one week: this choice is consistent with the publishing practices of the channels, which tend to have weekly rhythms, e.g.\ being more active on working days than during weekends~\cite{castaldo2021rhythms}.
Figure~\ref{fig:evolution_visuals_both} shows the evolution of normalized views $\Pi(t)$ as defined in \eqref{def:visualz_norm}: the measure of popularity exhibits irregular behavior, which is smoothed if we consider the time averages of normalized views
$\bar{\Pi}^{(i)}(T) = \frac1T\sum_{t=1}^T \Pi ^{(i)}(t), $ for $ i\in\mathcal{C}$,
as in the right panel.
\begin{figure}[h]
    \centering
    \includegraphics[width=0.9\columnwidth]{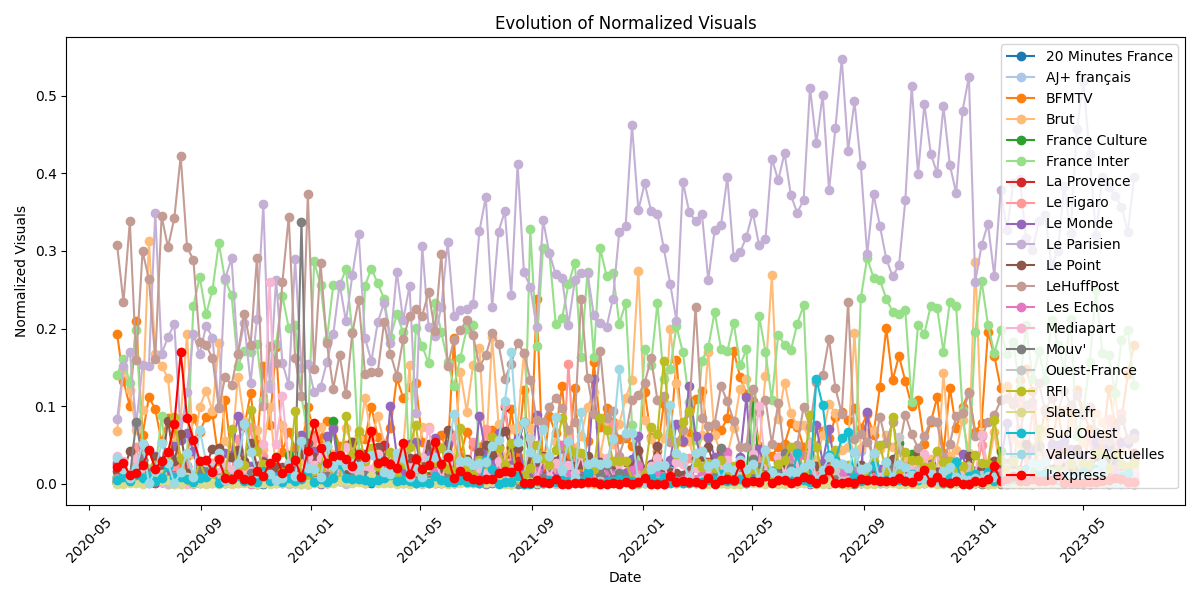}
    \includegraphics[width=0.9\columnwidth]{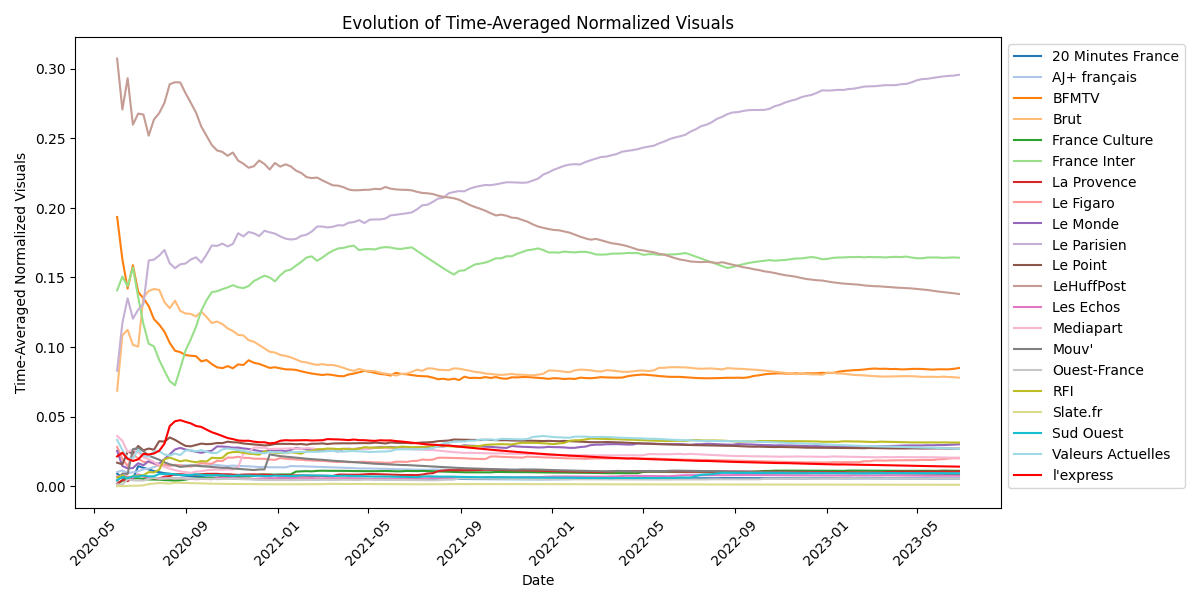}
    \caption{\small Evolution of the popularity measure, i.e., the normalized number of views (left), and its time averages (right) for the main French news YouTube channels.}
        \label{fig:evolution_visuals_both}
\end{figure}
Comparing the two diagrams, we observe that popularity exhibits a behavior that is compatible with stationarity/time-averaging: most channels, including the prominent \textit{France Inter}, \textit{Brut}, and \textit{BFMTV}, show convergence of the time averages. Two channels, \textit{Le Parisien} and \textit{Le HuffPost}, exhibit increasing/decreasing trends: such trends may depend on exogenous forces or on the evolution in time of the channel qualities, which are not accounted for in our model.

\section{Concluding remarks}\label{sec:conclusions}
In this work, we proposed a stochastic model for the evolution of attention and popularity in social networks. Empirical data from YouTube provide evidence that real popularity trajectories display behaviors consistent with ergodic stochastic dynamics.
The analysis of the model showed that it exhibits qualitatively different regimes, depending on the strength of the influence from the quality of content.
The numerical experiments illustrated the different asymptotic behaviors predicted by the theory and clarified the role of the model parameters in the distribution of popularity. 

Several directions remain open for future research. A first extension concerns the introduction of time-varying qualities, in order to account for the fact that the attractiveness of a channel may evolve over time. A second direction of our current research regards the study of alternative exposure mechanisms and control techniques for anchoring attention to high-quality content. 

\section{Acknowledgments}
The authors would like to thank Yhaya Bajdadi and Titouan Vial for their help in collecting and analyzing the YouTube dataset, as well as Gilles Bastin, Emmanuel Marty, and Jérôme Pacouret for compiling the list of French news outlets. 

During the preparation of this manuscript, the GPT-5.4 system was used for proofreading, editing, and drafting parts of the simulation code. The scientific content, analyses, interpretations, and conclusions were developed entirely by the authors.
\bibliographystyle{IEEEtran}
\bibliography{references}
\appendix
\subsection{Proof of Theorem~\ref{thm:no_quality_all_alpha}}\label{app:A}
We prove the three claims in order.

{(\emph{Absorbing set})}
For every $j\in\mathcal{C}$, the state $z^{\star,j}$ represents the scenario in which all users engage only with channel $j$.
Then, for the channel $j\in\mathcal{C}$ and for all $v\in\mathcal{V}$, we have
\begin{align*}
&\mathbb{P}\left(Y_v^{(j)}(t+1)=1 \mid Y(t)=z^{\star,j}\right)\\
&\quad
= \alpha_v\sum_{w\in\mathcal{V}}P_{vw}\cdot 1 + \beta_v\cdot 1 = \alpha_v+\beta_v = 1.
\end{align*}
For all the other channels $i\neq j$, we get
\begin{align*}
&\mathbb{P}\left(Y_v^{(i)}(t+1)=1 \mid Y(t)=z^{\star,j}\right)\\
&\quad= \alpha_v\sum_{w\in\mathcal{V}}P_{vw}\cdot 0 + \beta_v\cdot 0 = 0
\qquad \forall v\in\mathcal{V}.
\end{align*}
This implies that
$
W_{z^{\star,j}\leftarrow z^{\star,j}}=1,
$ and $
W_{z'\leftarrow z^{\star,j}}=0,$ $\forall z'\neq z^{\star,j}.
$
We conclude that all states in $\mathcal{Z}^{\star}$ are absorbing.

Moreover, if the starting state is $z\notin\mathcal{Z}^{\star}$, then from \eqref{eq:theta} at least one entry of $z$ can change with positive probability, and therefore
$
W_{z\leftarrow z}<1.
$
Hence, the states in $\mathcal{Z}^{\star}$ are the only absorbing states.

\smallskip
{(\emph{Reachability})}
It is immediate to observe that if the column vector $Y^{(j)}(0)$ has all entries equal to zero, then the state $z^{\star,j}$ cannot be reached, because channel $j$ has no initial positive engagement and there is no quality term influencing the dynamics.

We now prove that $\sum_{v\in\mathcal{V}}Y_v^{(j)}(0)\geq 1$ is sufficient for $z^{\star,j}$ to be reachable with positive probability.
Indeed, if $\sum_{v\in\mathcal{V}}Y_v^{(j)}(0)\geq 1$, then $\Pi^{(j)}(0)>0$ and, since $\beta_v>0$ for all $v$, we get
\begin{align*}
&\mathbb{P}\left(Y_v^{(j)}(1)=1 \mid Y(0)\right)
\\
&\qquad=
\alpha_v\sum_{w\in\mathcal{V}}P_{vw}Y_w^{(j)}(0)
+\beta_v\Pi^{(j)}(0)
>0.
\end{align*}

Notice that for all $i\neq j$ we have
\begin{align*}
X_v^{(i)}(1)
&=
\alpha_v\sum_{u\in\mathcal{V}}P_{vu}Y_u^{(i)}(0)+\beta_v\Pi^{(i)}(0)
\\
&\le
\alpha_v\sum_{u\in\mathcal{V}}P_{vu}+(1-\alpha_v)(1-\Pi^{(j)}(0))
\\
&\leq 1-\beta_v\Pi^{(j)}(0),
\end{align*}
where we used $\beta_v=1-\alpha_v$ because $\gamma_v=0$. Therefore,
$$
\mathbb{P}\left(Y_v^{(i)}(1)=0 \mid Y(0)\right)
=1-X_v^{(i)}(1)
\geq \beta_v\Pi^{(j)}(0)>0.
$$
Let now $w\in\mathcal{V}$ be such that $Y_w^{(j)}(0)=1$.
Since $P$ is strongly connected, for every node $v\in\mathcal{V}$ there exists a directed path
$
w=v_0 \to v_1 \to \cdots \to v_L=v
$
such that $P_{v_{k+1}v_k}>0$ for all $k=0,\dots,L-1$.
By induction along the path, one shows that there is a strictly positive probability that, at successive times, each node $v_k$ engages with channel $j$ and does not engage with the other channels. Hence, after finitely many steps, every node can reach the configuration
$
Y_v^{(j)}=1,\  Y_v^{(i)}=0, \ \forall i\neq j,
$
with strictly positive probability. Therefore, the absorbing state $z^{\star,j}$ is reachable.

\smallskip
{(\emph{Convergence})}
Let $z$ be any initial state. If $z\in\mathcal{Z}^{\star}$, there is nothing to prove. If $z\notin\mathcal{Z}^{\star}$, then either:
\begin{itemize}
\item there exists at least one channel $j$ such that $\sum_{v\in\mathcal{V}}Y_v^{(j)}(0)\geq 1$, and then by point (ii) the absorbing state $z^{\star,j}$ is reachable from $z$ with positive probability;
\item or $Y_v^{(i)}(0)=0$ for all $v\in\mathcal{V}$ and all $i\in\mathcal{C}$, in which case, by the convention in \eqref{eq:rankingstoch}, $\Pi^{(j)}(0)=1/m  >0$ for all $j$, and again every $z^{\star,j}$ is reachable with positive probability.
\end{itemize}
Hence, from every non-absorbing state, at least one absorbing state is reachable. 
Therefore, the Markov chain is absorbed in $\mathcal{Z}^{\star}$ with probability one. 

\section{Proofs for the winning probabilities}
\label{app:winning_probabilities}

\subsection{Proof of Proposition~\ref{thm:two_community_network}}
\label{app:two_community_network}

Since $\alpha_v=1-\beta$ for every $v\in\mathcal V$, we have
$A=(1-\beta)I$ and hence
$
AP=(1-\beta)P.
$
Then, the Perron eigenvalue introduced in
Theorem~\ref{thm:winner_decomposition} is
$\lambda=1-\beta$, while its normalized left Perron eigenvector
$\pi$ coincides with the stationary distribution of $P$.
For the two-community matrix $P$, this distribution is constant within each
community. Standard calculations lead to 
$$
\pi_v=
\begin{cases}
\dfrac{b}{(a+b)n_1}, & v\in\mathcal V_1,\\
\dfrac{a}{(a+b)n_2}, & v\in\mathcal V_2.
\end{cases}
$$
Then for every state $z\in\{0,1\}^{\mathcal{C}\times\mathcal{V}}$ and channel $i$,
$$
C_i(z)
=
\sum_{v\in\mathcal V}\pi_v z_{v,i}
=
\frac{b}{a+b}\frac{k_{1,i}(z)}{n_1}
+
\frac{a}{a+b}\frac{k_{2,i}(z)}{n_2},
$$
which proves~\eqref{eq:core_periphery_network_score}.

We now prove part~(ii). We first observe that
\begin{align*}
C_i(z)-\frac{k_i(z)}{n}
&=
\frac{bn_2-an_1}{n(a+b)}
\left(
\frac{k_{1,i}(z)}{n_1}
-
\frac{k_{2,i}(z)}{n_2}
\right).
\end{align*}
Moreover, for $S(z)>0$,
$
\Pi^{(i)}(z)={k_i(z)}/{S(z)}
$
and
$$
\Pi^{(i)}(z)-C_i(z)
=
\left(
\frac{k_i(z)}{S(z)}-\frac{k_i(z)}{n}
\right)
-
\left(
C_i(z)-\frac{k_i(z)}{n}
\right)
$$

For $v\in\mathcal V_r$, the conditional success
probability of $Y_v^{(i)}(t+1)$, given $Y(t)=z$, is
$
(1-\beta)\sum_{w\in\mathcal V}P_{vw}z_{w,i}
+
\beta\Pi^{(i)}(z).
$
By the definition of $P$, we obtain
$$
\sum_{w\in\mathcal V}P_{vw}z_{w,i}
=
\sum_{s=1}^2
\frac{B_{rs}}{n_s}
\sum_{w\in\mathcal V_s}z_{w,i}
=
\sum_{s=1}^2B_{rs}\frac{k_{s,i}(z)}{n_s}.
$$
Furthermore, $\Pi^{(i)}(z)$ depends on $z$ only through its aggregate
counts and is equal to $p_i(k(z))$. Consequently, the conditional law of
the next aggregate state depends on the current microscopic state only
through $k(z)$. Hence $K(t)$ is itself a Markov chain.

In particular, conditionally on $K(t)=k$, all nodes in community $r$ have
the same success probability
$$
\theta_{r,i}(k)
=
(1-\beta)\sum_{s=1}^2B_{rs}\frac{k_{s,i}}{n_s}
+
\beta p_i(k).
$$
Since the variables $Y_v^{(i)}(t+1)$ are conditionally independent across
$v$, it follows that
$
K_{r,i}(t+1)
\sim
\operatorname{Binomial}
\left(n_r,\theta_{r,i}(k)\right),
$
proving~\eqref{eq:core_periphery_theta}.

It remains to derive the expression for the hitting probability.
Notice that both terms on the right-hand side of \eqref{eq:winner_decomposition} depend on $Y(t)$ only through $K(t)$. More
precisely, for an aggregate state $k$,
$$
\Pi^{(i)}(Y(t))-C_i(Y(t))
=
r_i(k),
$$
where
$$
r_i(k)
=
p_i(k)
-
\frac{b}{a+b}\frac{k_{1,i}}{n_1}
-
\frac{a}{a+b}\frac{k_{2,i}}{n_2}.
$$

Let $\widehat Q$ denote the transition matrix of the aggregate chain
restricted to its transient states. By the standard fundamental-matrix
formula for finite absorbing Markov chains,
$$
\mathbb E_k
\left[
\sum_{t=0}^{\tau_{\mathcal Z^\star}-1}r_i(K(t))
\right]
=
\left[
(I-\widehat Q)^{-1}r_i
\right](k),
$$
which, applied to the initial state $k(z)$, gives
~\eqref{eq:core_periphery_closed_hitting}.

\subsection{Proof of Theorem~\ref{thm:ergodicity}}
We begin the proof by proving the irreducibility of the chain. 
We know that
\begin{equation*}
W_{z'\leftarrow z} = \prod_{i=1}^m \prod_{v \in \mathcal{S}^{(i)}({z'})} (\theta_z)^{i}_{v} \prod_{v \notin \mathcal{S}^{(i)}({z'})} \big(1 -(\theta_z)^{i}_{v}\big),
\end{equation*}
where 
$
(\theta_z)^{i}_{v}
$
is defined by \eqref{eq:theta}.
By hypothesis $\gamma_v>0 \ \forall v$ and $q^{(i)}>0 \ \forall i$, then $(\theta_z)^{i}_{v}>0, \ \forall v, \ \forall i, \ \forall z$. At the same time $q^{(i)}<1 \ \forall i$ implies that 
\begin{align*}
(\theta_z)^{i}_{v}& \leq \alpha_v+\beta_v+\gamma_vq^{(i)}<1.
\end{align*}
From the latter observations, we can conclude that 
$
W_{z'\leftarrow z}>0
$ for all $z, z'$.
This means that it is possible to go from one state to any other state at any time.
We conclude that the chain is irreducible and aperiodic. Since the Markov chain has a finite state space, irreducibility implies the existence and uniqueness of an invariant distribution, and together with aperiodicity it yields convergence in distribution (see \cite{Borkar1995}). We conclude that \eqref{eq:ergodicity} holds almost surely.

\subsection{Proof of Proposition~\ref{prop:mix_time}}
\label{app:B}

To prove the statement, we use a Markovian coupling
procedure \cite{levin2017markov}. More precisely, given two current states
$Y(t)=y$ and $Y'(t)=z$, we construct a joint
transition kernel
\begin{align*}
&\widetilde W_{(y',z')\leftarrow(y,z)}
\\
&=
\mathbb{P}\big(
Y(t+1)=y',Y'(t+1)=z'
\mid Y(t)=y,Y'(t)=z
\big),
\end{align*}
whose marginals coincide with the original
transition kernel $W$. 
The coupling is constructed as follows.

For each pair $(v,i)\in\mathcal{V}\times\mathcal{C}$,
let $U_v^{(i)}(t)\sim\mathrm{Unif}(0,1)$.
If $U_v^{(i)}(t)\leq\gamma_v$, then a Bernoulli
random variable with parameter $q^{(i)}$
is generated, i.e.,
$B_v^{(i)}(t)\sim\mathrm{Ber}(q^{(i)})$,
and we set
$
Y_v^{(i)}(t+1)
=
Y_v'^{(i)}(t+1)
=
B_v^{(i)}.
$
On the other hand, if $U_v^{(i)}(t)>\gamma_v$, then,
for $\gamma_v<1$, we define
\begin{align*}
r_v^{(i)}(y)
&=
\frac{
\alpha_v\sum_{w\in\mathcal{V}}P_{vw}y_w^{(i)}
+\beta_v\Pi^{(i)}(y)}
{1-\gamma_v},\\
r_v^{(i)}(z)
&=
\frac{
\alpha_v\sum_{w\in\mathcal{V}}P_{vw}z_w^{(i)}
+\beta_v\Pi^{(i)}(z)}
{1-\gamma_v}.
\end{align*}
Notice that
$r_v^{(i)}(y),r_v^{(i)}(z)\in[0,1]$,
since $\alpha_v+\beta_v=1-\gamma_v$.
%
Generate $\Xi_v^{(i)}(t)\sim\mathrm{Unif}(0,1)$,
independently across pairs and independently
of the other random variables, and set
\begin{align*}
Y_v^{(i)}(t+1)
&=
\mathbf{1}_{\{
\Xi_v^{(i)}\leq r_v^{(i)}(y)
\}},\ Y_v'^{(i)}(t+1)
&=
\mathbf{1}_{\{
\Xi_v^{(i)}\leq r_v^{(i)}(z)
\}}.
\end{align*}
Then, we have
$\mathbb{P}\big(
Y_v^{(i)}(t+1)=1
\mid Y(t)=y,Y'(t)=z
\big)
=
\gamma_v q^{(i)}
+(1-\gamma_v)r_v^{(i)}(y)
=\theta_v^{(i)}(y),
$
where $\theta_v^{(i)}(y)$ is the transition
probability defined in Eq.~(4).
Similarly,
$\mathbb{P}\big(
Y_v'^{(i)}(t+1)=1
\mid Y(t)=y,Y'(t)=z
\big)=\theta_v^{(i)}(z),
$
from which we can conclude that the joint kernel defines a valid Markovian coupling
of the two chains.

Observe that, if
$U_v^{(i)}(t)\leq\gamma_v$ for every
$(v,i)\in\mathcal{V}\times\mathcal{C}$,
the two chains coincide at time $t+1$.
Consequently,
\begin{align*}
&\mathbb{P}\big(
Y(t+1)=Y'(t+1)
\mid Y(t)=y,Y'(t)=z
\big)
\geq
(\gamma_{\min})^{nm}.
\end{align*}
Using the coupling inequality,
\begin{align*}
&\big\|W(y,\cdot)-W(z,\cdot)\big\|_{\mathrm{TV}}\\
&\quad\leq
1-\mathbb{P}\big(
Y(t+1)=Y'(t+1)
\mid Y(t)=y,Y'(t)=z
\big)\\
&\quad\leq1-(\gamma_{\min})^{nm}.
\end{align*}
This leads to the following bound on the
Dobrushin ergodic coefficient:
\[
\delta(W):=
\sup_{y,z\in\{0,1\}^{\mathcal{V}\times\mathcal{C}}}
\big\|W(y,\cdot)-W(z,\cdot)\big\|_{\mathrm{TV}}
\leq1-(\gamma_{\min})^{nm}.
\]
By iterating the contraction inequality
and using the invariance of $\nu$, we obtain
\begin{align*}
&\max_{x\in\{0,1\}^{\mathcal{V}\times\mathcal{C}}}
\big\|W^t(x,\cdot)-\nu\big\|_{\mathrm{TV}}\leq\delta(W)^t
\leq
\mathrm{e}^{-(\gamma_{\min})^{nm}t,}
\end{align*}
By imposing
$
\exp\big(-(\gamma_{\min})^{nm}t\big)
\leq\varepsilon,
$
we obtain
$
t\geq
\frac{1}{(\gamma_{\min})^{nm}}
\log({1}/{\varepsilon}).
$
Since $t$ is integer-valued, it follows that
$
\tau_{\mathrm{mix}}(\varepsilon)
\leq
\left\lceil
\frac{1}{(\gamma_{\min})^{nm}}
\log\frac{1}{\varepsilon}
\right\rceil,
$
which concludes the proof.

\subsection{Proof of Theorem~\ref{thm:exit}}
We begin with the following preliminary result.
\begin{lemma}\label{lemma:E}
Let $\mathcal F_t$ be the $\sigma$-algebra generated by the history of $Y$ up to time $t$. Fix $i\in\mathcal C$ and define the following events
\begin{gather*}
A_i:=\Big\{\sum_{v\in\mathcal V}Y_v^{(i)}(t+1)<n \Big\},\\
B_i:=\bigcup_{j\neq i}\bigcup_{v\in\mathcal V}\{Y_v^{(j)}(t+1)=1\},\quad E_i:=A_i\cup B_i.
\end{gather*}
Then,
$$\mathds{1}_{\{Y(t)=z^{\star,i}\}}
\mathbb{P}\left(Y(t+1)\neq z^{\star,i}\mid \mathcal{F}_t\right)
=
\mathds{1}_{\{Y(t)=z^{\star,i}\}}
\mathbb{P}\left(E_i\mid \mathcal{F}_t\right).
$$\end{lemma}

\begin{proof}
By the definition of monopoly state, $
Y(t+1)=z^{\star,i}
$
if and only if all users keep positive engagement on channel $i$ and no user positively engages with any other channel. Equivalently, we write 
$
Y(t+1)\neq z^{\star,i}
$
if and only if either at least one user has $Y_v^{(i)}(t+1)=0$, or there exists $j\neq i$ and $v\in\mathcal V$ such that $Y_v^{(j)}(t+1)=1$. This is precisely the event $E_i$.
Taking conditional probabilities given $\mathcal F_t$ yields the claim.
\end{proof}

\smallskip
Fix $i\in\mathcal C$ and assume that $Y(t)=z^{\star,i}$. Then, by \eqref{eq: stoch_model_simpl}, for every $v\in\mathcal V$ we have
\begin{align*}
X_v^{(i)}(t+1)
&=
\alpha_v\sum_{w\in\mathcal V}P_{vw}\cdot 1+\beta_v\cdot 1+\gamma_v q^{(i)}\\
&=\alpha_v+\beta_v+\gamma_v q^{(i)}
=1-\gamma_v(1-q^{(i)}),
\end{align*}
whereas for every $j\neq i$,
$
X_v^{(j)}(t+1)=\gamma_v q^{(j)}.
$

Let $\mathcal F_t$ be the $\sigma$-algebra generated by the history up to time $t$.
By Lemma~\ref{lemma:E}
\begin{align*}
&\mathbb P\left(Y(t+1)\neq z^{\star,i}\mid \!\mathcal F_t\right)\mathds{1}_{\{Y(t)=z^{\star,i}\}}
\!=
\mathbb P(E_i\mid\mathcal \!F_t)\mathds{1}_{\{Y(t)=z^{\star,i}\}}.
\end{align*}

We first derive a lower bound. Since
$
A_i\subseteq E_i,
$
we get
$
\mathbb P(E_i\mid\mathcal F_t)\geq \mathbb P(A_i\mid\mathcal F_t).
$
Moreover, conditioned on $Y(t)=z^{\star,i}$, the variables $\{Y_v^{(i)}(t+1)\}_{v\in\mathcal V}$ are independent Bernoulli random variables with success probabilities $X_v^{(i)}(t+1)=1-\gamma_v(1-q^{(i)})$. Hence
\begin{align*}
&\mathbb P(A_i\mid\mathcal F_t)
=
1-\prod_{v\in\mathcal V}X_v^{(i)}(t+1)\\
&=
1-\prod_{v\in\mathcal V}\left(1-\gamma_v(1-q^{(i)})\right)\ge
1-\left(1-\gamma_{\min}(1-q^{(i)})\right)^{n }.
\end{align*}
If $\gamma_{\min}$ is sufficiently small so that
$
n \gamma_{\min}(1-q^{(i)})<1,
$
then
\begin{align*}
\mathbb P(A_i\mid\mathcal F_t)
&\ge
1-e^{-n \gamma_{\min}(1-q^{(i)})}\ge
\frac12\,n \gamma_{\min}(1-q^{(i)}).
\end{align*}
Thus
$
\mathbb P\left(Y(t+1)\neq z^{\star,i}\mid Y(t)=z^{\star,i}\right)
\geq c_1^{(i)}\gamma_{\min},
$
with
$
c_1^{(i)}=\frac12\,n (1-q^{(i)}).
$

We now derive an upper bound. By the union bound and Lemma~\ref{lemma:E},
$
\mathbb P(E_i\mid\mathcal F_t)
\le
\mathbb P(A_i\mid\mathcal F_t)+\mathbb P(B_i\mid\mathcal F_t).
$
For the first term, again by the union bound,
\begin{align*}
\mathbb P(A_i\mid\mathcal F_t)
&\le
\sum_{v\in\mathcal V}\mathbb P\left(Y_v^{(i)}(t+1)=0\mid\mathcal F_t\right)\\
\quad&=
\sum_{v\in\mathcal V}\left(1-X_v^{(i)}(t+1)\right)=
\sum_{v\in\mathcal V}\gamma_v(1-q^{(i)})\\
&\le
n \,\gamma_{\max}(1-q^{(i)}).
\end{align*}
The second term can be estimated as follows
\begin{align*}
\mathbb P(B_i\mid\mathcal F_t)
&\le
\sum_{j\neq i}\sum_{v\in\mathcal V}\mathbb P\left(Y_v^{(j)}(t+1)=1\mid\mathcal F_t\right)\\
&=
\sum_{j\neq i}\sum_{v\in\mathcal V}X_v^{(j)}(t+1)=
\sum_{j\neq i}\sum_{v\in\mathcal V}\gamma_v q^{(j)}\\
&\le
(m  -1)n \gamma_{\max}q_{\max}.
\end{align*}
Combining the two bounds yields
\begin{align*}
&\mathbb P\left(Y(t+1)\neq z^{\star,i}\mid Y(t)=z^{\star,i}\right)\\
&\quad\le
n \gamma_{\max}(1-q^{(i)})+
(m  -1)n \gamma_{\max}q_{\max}\\
&\quad=
c_2^{(i)}\gamma_{\max},
\end{align*}
with
$
c_2^{(i)}
=
n \big((1-q^{(i)})+(m-1)q_{\max}\big),
$ concluding the proof.

\subsection{Proof of Corollary~\ref{corol:X}}
From Theorem~\ref{thm:ergodicity}, the Markov chain $\{Y(t)\}$ converges in distribution to $Y_\infty \sim \nu$. 
Since $X(t)=f_X(Y(t))$ and $\Pi(t)=f_\Pi(Y(t))$ are continuous and bounded deterministic functions of $Y(t)$, by the continuous mapping theorem (see \cite[Theorem~2.3]{Billingsley1999}) we have that $\{X(t)\}_{t\in\integernonnegative}$ and $\{\Pi(t) \}_{t\in\integernonnegative}$ converge in distribution to $X_\infty=f_X(Y_\infty)$ and $\Pi_\infty=f_\Pi(Y_\infty)$, respectively, thus proving the first statement. 
Since $\{Y(t)\}$ is ergodic, applying \eqref{eq:ergodicity} with $f=f_X$ and $f=f_\Pi$ yields the second statement.

\subsection{Proof of Theorem~\ref{thm:metastability}} 
Let $
\beta_{\min}:=\min_{v\in\mathcal V}\beta_v>0$,
$\gamma_{\max}:=\max_{v\in\mathcal V}\gamma_v$,
 $q_{\max}:=\max_{i\in\mathcal C}q^{(i)}$.
Also, let
$
p_t:=\mathbb P\left(Y(t)\notin\mathcal{Z}^{\star}\right),
$
and let $\mathcal F_t$ be the $\sigma$-algebra generated by the history up to time $t$.
By conditioning on the complementary events $\{Y(t)\in\mathcal{Z}^{\star}\}$ and $\{Y(t)\notin\mathcal{Z}^{\star}\}$, we obtain
\begin{align}
p_{t+1}
&=
\mathbb P\left(Y(t+1)\notin\mathcal{Z}^{\star}\right)\nonumber\\
&=
\mathbb E\!\left[
\mathbf 1_{\{Y(t)\in\mathcal{Z}^{\star}\}}
\mathbb P\left(Y(t+1)\notin\mathcal{Z}^{\star}\mid\mathcal F_t\right)
\right]\nonumber\\
&\quad+
\mathbb E\!\left[
\mathbf 1_{\{Y(t)\notin\mathcal{Z}^{\star}\}}
\mathbb P\left(Y(t+1)\notin\mathcal{Z}^{\star}\mid\mathcal F_t\right)
\right].
\label{eq:split-U-new}
\end{align}
If $Y(t)=z^{\star,i}$ for some $i\in\mathcal C$, then Theorem~\ref{thm:exit} gives
\begin{align*}
&\mathbb P\left(Y(t+1)\notin\mathcal{Z}^{\star}\mid Y(t)=z^{\star,i}\right)
\\&\quad\leq
\mathbb P\left(Y(t+1)\neq z^{\star,i}\mid Y(t)=z^{\star,i}\right)
\leq c_2\gamma_{\max},
\end{align*}
where the constant can be taken uniformly over channels as
$
c_2=n \left(1+(m  -1)q_{\max}\right).
$
Hence
\begin{align}
&\mathbb E\!\left[
\mathbf 1_{\{Y(t)\in\mathcal{Z}^{\star}\}}
\mathbb P\left(Y(t+1)\notin\mathcal{Z}^{\star}\mid\mathcal F_t\right)
\right]\nonumber\\
&\qquad\le
c_2\gamma_{\max}\,\mathbb P\left(Y(t)\in\mathcal{Z}^{\star}\right)
=
c_2\gamma_{\max}(1-p_t).
\label{eq:bound-inside-U}
\end{align}

We now estimate the contribution of states outside $\mathcal{Z}^{\star}$.
Assume that $Y(t)\notin\mathcal{Z}^{\star}$, and let $i_t\in\mathcal C$ be an
$\mathcal F_t$-measurable index such that
$
\Pi^{(i_t)}(t):=\max_{i\in\mathcal C}\Pi^{(i)}(t).
$
Since $\sum_{i\in\mathcal C}\Pi^{(i)}(t)=1$, by the pigeonhole principle
$
\Pi^{(i_t)}(t)\geq {1}/{m  }.
$
Hence, for every $v\in\mathcal V$,
\begin{align*}
\mathbb P\left(Y_v^{(i_t)}(t+1)=1\mid\mathcal F_t\right)
&\geq
\beta_v \Pi^{(i_t)}(t)\geq
{\beta_{\min}}/{m  }.
\end{align*}
Moreover, for every $j\neq i_t$, since
$
\Pi^{(j)}(t)\leq 1-\Pi^{(i_t)}(t)\leq 1-{1}/{m  },
$
we have
$
1-\Pi^{(j)}(t)\geq {1}/{m  }.
$
Therefore, for every $v\in\mathcal V$ and $j\neq i_t$,
\begin{align*}
\mathbb P\left(Y_v^{(j)}(t+1)=0\mid\mathcal F_t\right)
&=
1-X_v^{(j)}(t+1)\\
&\ge
\beta_v\bigl(1-\Pi^{(j)}(t)\bigr)\ge
\frac{\beta_{\min}}{m  }.
\end{align*}

Now define the event
$
G_t:=\{Y(t+1)=z^{\star,i_t}\}.
$
Since $G_t\subseteq \{Y(t+1)\in\mathcal{Z}^{\star}\}$, we obtain
$
\mathbb P\left(Y(t+1)\in\mathcal{Z}^{\star}\mid\mathcal F_t\right)
\ge
\mathbb P(G_t\mid\mathcal F_t).
$
Conditioned on $\mathcal F_t$, the random variables
$\{Y_v^{(j)}(t+1)\}_{v\in\mathcal V,\ j\in\mathcal C}$
are independent Bernoulli. Hence, on the event
$\{Y(t)\notin\mathcal{Z}^{\star}\}$,
{\small{\begin{align*}
&\mathbb P(G_t\mid\mathcal F_t)
\\&=
\prod_{v\in\mathcal V}
\mathbb P\left(Y_v^{(i_t)}(t+1)=1\mid\mathcal F_t\right)
\prod_{j\neq i_t}\prod_{v\in\mathcal V}
\mathbb P\left(Y_v^{(j)}(t+1)=0\mid\mathcal F_t\right)\\
&\ge
\left(\frac{\beta_{\min}}{m  }\right)^{n }
\left(\frac{\beta_{\min}}{m  }\right)^{(m  -1)n }=
\left(\frac{\beta_{\min}}{m  }\right)^{m  n }=:b_0.
\end{align*}}}
Then,
$
\mathbb P\left(Y(t+1)\notin\mathcal{Z}^{\star}\mid\mathcal F_t\right)
\leq 1-b_0$
on
$\{Y(t)\notin\mathcal{Z}^{\star}\},
$
and
\begin{align}
&\mathbb E\!\left[
\mathbf 1_{\{Y(t)\notin\mathcal{Z}^{\star}\}}
\mathbb P\left(Y(t+1)\notin\mathcal{Z}^{\star}\mid\mathcal F_t\right)
\right]\le
(1-b_0)p_t.
\label{eq:bound-outside-U}
\end{align}

Combining \eqref{eq:split-U-new}, \eqref{eq:bound-inside-U}, and \eqref{eq:bound-outside-U}, we get
$$
p_{t+1}\leq c_2\gamma_{\max}(1-p_t)+(1-b_0)p_t
\le
(1-b_0)p_t+c_2\gamma_{\max}.
$$
Iterating this inequality yields
\begin{align*}
p_t
&\le
(1-b_0)^t p_0
+
c_2\gamma_{\max}\sum_{k=0}^{t-1}(1-b_0)^k\\
&\le
(1-b_0)^t+\frac{c_2\gamma_{\max}}{b_0}.
\end{align*}
This is equivalent to stating that 
 for all $t\in\mathbb Z_{\geq 0}$ and for any initial condition, it holds
\begin{equation}\label{eq:finite-time-U}
\mathbb P\left(Y(t)\in\mathcal{Z}^{\star}\right)
\ge
1-\frac{c_2}{b_0}\gamma_{\max}-(1-b_0)^t.
\end{equation}
Ergodicity ensures that, for any initial condition,
$
\frac{1}{T}\sum_{t=1}^{T}\mathbf 1_{\{Y(t)\in\mathcal{Z}^{\star}\}}$
convereges to
$\nu_\gamma(\mathcal{Z}^{\star}).
$ a.s.\ as $T\to\infty$.
By dominated convergence,
$$
\nu_\gamma(\mathcal{Z}^{\star})
=
\lim_{T\to\infty}
\frac{1}{T}\sum_{t=1}^{T}\mathbb P\left(Y(t)\in\mathcal{Z}^{\star}\right),
$$
and, using \eqref{eq:finite-time-U}, we obtain
\begin{align*}
\nu_\gamma(\mathcal{Z}^{\star})
&\ge
\liminf_{T\to\infty}
\frac{1}{T}\sum_{t=1}^{T}
\left(1-\frac{c_2}{b_0}\gamma_{\max}-(1-b_0)^t\right)\\
&=
1-\frac{c_2}{b_0}\gamma_{\max},
\end{align*} 
thereby concluding the main statement.

\end{document}